\pdfoutput=1

\documentclass[10pt, a4paper]{article}

\usepackage[utf8]{inputenc}
\usepackage[T1]{fontenc}
\usepackage{lmodern} % For higher quality fonts

\usepackage[margin=2cm]{geometry}

\usepackage[activate={true,nocompatibility},final,tracking=true,
kerning=true,spacing=true,factor=1100,stretch=10,shrink=10]{microtype}
\SetTracking{encoding={*}, shape=sc}{40}

\usepackage{ellipsis}         % improves whitespacing around "..."s
\usepackage[abbreviations,british]{foreign} % for better typing i.e., e.g. etc. 

\usepackage{xcolor}
\usepackage{hyperref}
\usepackage[xcolor, hyperref, notion, paper]{knowledge}

\definecolor{darkmidnightblue}{rgb}{0.0, 0.2, 0.4}
\definecolor{persianplum}{rgb}{0.44, 0.11, 0.11}

\hypersetup{
  colorlinks  = true,        % Colours links instead of ugly boxes
  citecolor   = persianplum, % Colour of citations
  urlcolor    = persianplum, % Colour for external hyperlinks,
  linkcolor   = persianplum
}

\usepackage{amssymb} 
\usepackage{amsmath}
\usepackage{amsthm}
\usepackage{braket} % sets
\usepackage{xurl}

\usepackage[capitalise, noabbrev]{cleveref}

\newtheorem{thm}{Theorem}[section]
\newtheorem{definition}[thm]{Definition}
\newtheorem{fact}[thm]{Fact}
\newtheorem{lemma}[thm]{Lemma}

\newtheorem{corollary}[thm]{Corollary}
\newtheorem{exercise}[thm]{Exercise}
\newtheorem{claim}[thm]{Claim}

\usepackage{caption}

\usepackage{wasysym}

\usepackage{tikz}
\usetikzlibrary{hobby,backgrounds,calc,trees}
\usetikzlibrary{arrows, decorations.markings, shapes, calc}
\usepackage{float}

\usepackage{smartdiagram}

\usetikzlibrary{arrows}
\usetikzlibrary{decorations.pathreplacing,calligraphy}

\colorlet{jaune}{yellow!80!green}

\colorlet{vert}{green!45!black}

\colorlet{bleu}{blue!70!black}

\colorlet{rouge}{red!80!black}

\tikzstyle{minirond}=[draw,circle,minimum height=2mm,inner sep=0pt]
\tikzstyle{ptrond}=[draw,circle,minimum height=2mm]
\tikzstyle{medrond}=[draw,circle,minimum height=5mm]
\tikzstyle{rond}=[draw,circle,minimum height=7mm]
\tikzstyle{carre}=[draw,minimum width=6mm,minimum height=6mm]
\tikzstyle{blanc}=[draw=black!80!white,fill=white]
\tikzstyle{rouge}=[draw=red,fill=red!20!white]
\tikzstyle{vert}=[draw=green!80!black,fill=green!80!black!20!white]
\tikzstyle{jaune}=[draw=yellow!60!red,fill=yellow!60!red!30!white]
\tikzstyle{bleu}=[draw=blue,fill=blue!40!white]
\tikzstyle{gris}=[draw=black!80!white,fill=black!15!white]

\newcommand{\orcid}[1]{\href{https://orcid.org/#1}{\includegraphics[width=9pt]{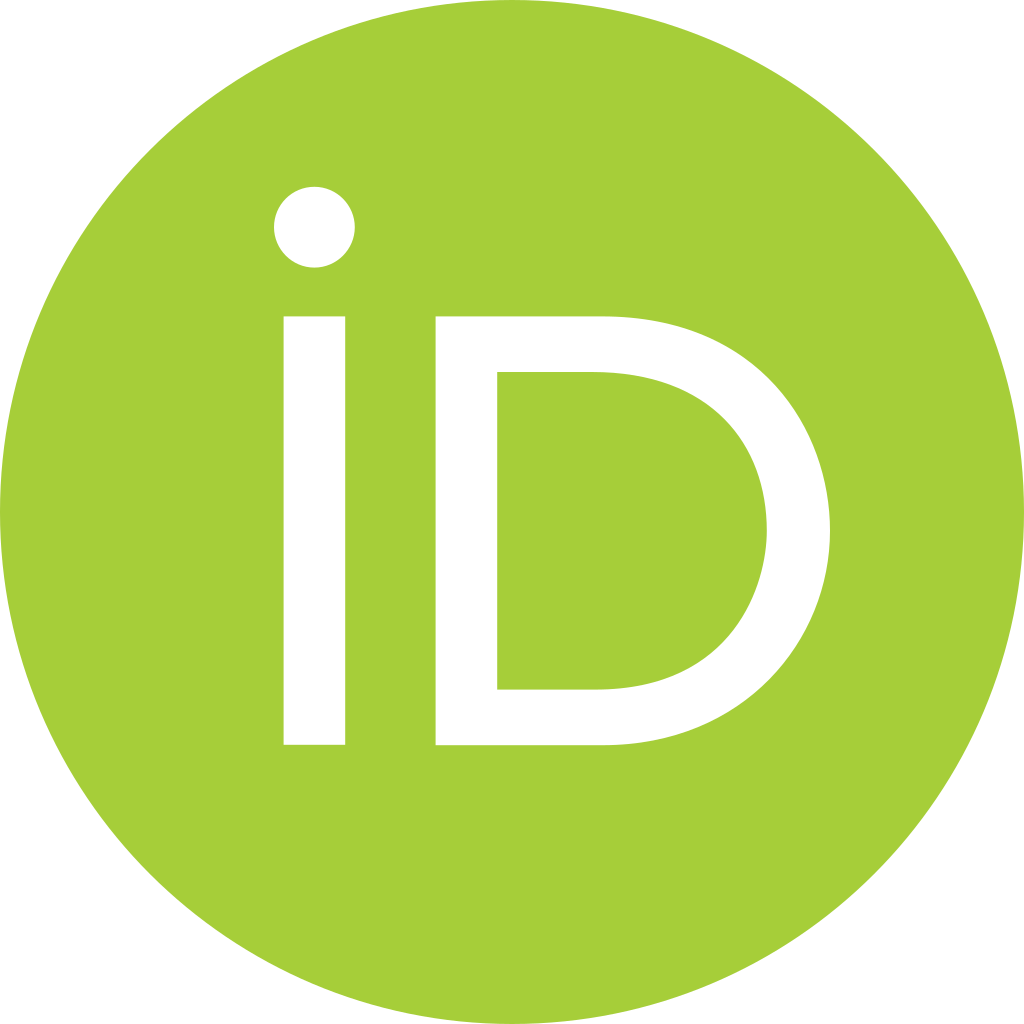}}}

\usepackage{todonotes}

\newcommand{\newnotion}[1]{\AP\intro{#1}}

\newcommand{\ModalitySpacer}{\,}
\newcommand*{\Fname}{\mathbf{F}}
\newcommand*{\Xname}{\mathbf{X}}
\newcommand*{\Halfname}{\mathbf{Half}}
\newcommand*{\Gname}{\mathbf{G}}
\newcommand*{\MajPname}{\mathbf{PM}}
\newcommand*{\MostFreqname}{\mathbf{MFL}}

\newcommand*{\F}[1]{\Fname\ModalitySpacer{}{#1}}
\renewcommand*{\P}[1]{\mathbf{P}\ModalitySpacer{}{#1}}
\newcommand*{\X}[1]{\mathbf{X}\ModalitySpacer{}{#1}}
\newcommand*{\U}[1]{\mathbf{U}\ModalitySpacer{}{#1}}
\newcommand*{\G}[1]{\Gname\ModalitySpacer{}{#1}}

\newcommand*{\MajP}[1]{{\MajPname}\ModalitySpacer{}{#1}}
\newcommand{\Half}[1]{\Halfname\ModalitySpacer{}{#1}}
\newcommand*{\MostFreq}[1][]{\ensuremath{{\MostFreqname}_{#1}}\ModalitySpacer{}}

\newcommand{\FO}{\mathrm{FO}}
\newcommand{\FOt}{\FO^2}
\newcommand{\FOtMaj}{\FOt_\Maj}

\newcommand{\LTL}{\mathrm{LTL}}

\newcommand{\LTLp}[1]{\mathrm{LTL}_{#1{}}}
\newcommand{\LTLF}{\LTLp{\Fname{}}}
\newcommand{\LTLFp}[1]{\LTLp{\Fname{}, #1{}}}
\newcommand{\LTLFMajP}{\LTLFp{\MajPname{}}}
\newcommand{\LTLFMostFreq}{\LTLFp{\MostFreqname}}
\newcommand{\LTLFHalf}{\LTLFp{\Halfname{}}}
\newcommand{\LTLprocent}{\LTL^{\%}}

\newcommand{\NExpTime}{\textsc{NExpTime}}
\newcommand{\NP}{\textsc{NP}}

\newcommand{\PSpace}{\textsc{PSpace}}

\newcommand{\deff}{\; := \;}
\newcommand{\deffnoindent}{:=}

\newcommand{\N}{{\mathbb{N}}}
\newcommand{\Z}{{\mathbb{Z}}} 
\newcommand{\setof}[1]{\{#1\}}

\newcommand{\APset}{\mathsf{AP}}

\newcommand{\str}[1]{{\mathfrak{#1}}}
\newcommand{\aword}{\str{w}}

\newcommand{\proj}[2]{\kl[proj]{\mathsf{proj}_{#1}{(#2)}}}

\newcommand{\mysharp}[3]{\kl[countingterm]{\#^{#1}_{#2}(#3)}}

\newcommand{\wht}{\mathit{wht}}
\newcommand{\shdw}{\mathit{shdw}}
\newcommand{\sigmashdw}{\tilde{\sigma}}
\newcommand{\alphashdw}{\tilde{\alpha}}
\newcommand{\betashdw}{\tilde{\beta}}

\newcommand{\phiinit}{\varphi_{\textit{init}}^{\textit{ex}\ref{ex:1}}}
\newcommand{\phiodd}{\kl[phiodd]{\varphi_{\textit{odd}}^{\textit{ex}\ref{ex:1}}}}
\newcommand{\phioddMFL}{\varphi_{\textit{odd}}^{\textit{MFL}}}

\newcommand{\psishadowy}{\kl[psishadowy]{\psi_{\textit{shadowy}}}}
\newcommand{\psishadowyMFL}{\kl[psishadowy]{\psi^{\textit{MFL}}_{\textit{shadowy}}}}

\newcommand{\philst}{\varphi_{\textit{last}}^{\textit{ex}\ref{ex:2}}}
\newcommand{\phistl}{\varphi_{\textit{stl}}^{\textit{ex}\ref{ex:2}}}
\newcommand{\phitransfer}[2]{\kl[phitransfer]{\varphi_{#1 \leadsto #2}^{\textit{trans}}}}
\newcommand{\psitrulySigmashadowy}{\kl[psitrulySigmashadowy]{\psi_{\textit{shadowy}}^{\textit{truly}-\Sigma}}}

\newcommand{\phiisequal}[2]{\kl[phiisequal]{\psi_{\# #1 {=} \# #2}}}

\newcommand{\minsky}{\kl[minsky]{\mathcal{A}}}
\newcommand{\mword}{\str{w}^{\minsky}}
\newcommand{\fancyq}{\mathsf{q}}

\newcommand{\psiminsky}{\psi_{\minsky}}
\newcommand{\psiminskyq}{\psiminsky^{\fancyq}}

\newcommand{\op}{\mathit{op}}
\newcommand{\fromq}[1]{\textit{from}_{#1}}
\newcommand{\toq}[1]{\textit{to}_{#1}}

\newcommand{\cfst}[1]{\mathit{fstVal}_{#1}}
\newcommand{\csnd}[1]{\mathit{sndVal}_{#1}}
\newcommand{\ifst}[1]{{\mathit{fstOP}}_{#1}}
\newcommand{\isnd}[1]{{\mathit{sndOP}}_{#1}}

\newcommand{\counterf}{\bar{\mathit{f}}}
\newcommand{\counters}{\bar{\mathit{s}}}

\newcommand{\run}{{\mathit{run}}}
\newcommand{\SigmaMinsky}{\Sigma_{\minsky}}
\newcommand{\psitrulySigmaMinsky}{\kl[psitrulySigmaMinsky]{\psi_{\textit{shadowy}}^{\textit{truly}-\SigmaMinsky}}}

\newcommand{\dehalf}[1]{\kl[dehalfication]{\mathsf{dehalf}(#1)}}

\newcommand{\kripkep}[1]{\mathcal{K}_{#1}}

\let\succ\undefined
\newcommand{\succ}{\mathit{succ}}
\newcommand{\Maj}{\mathsf{M}}
\newcommand{\HalfQ}{\underline{\mathsf{Half}}}
\newcommand{\first}{\underline{\mathit{first}}}
\newcommand{\second}{\underline{\mathit{second}}}
\newcommand{\sectolast}{\underline{\mathit{sectolast}}}
\newcommand{\last}{\underline{\mathit{last}}}
\newcommand{\distr}{\underline{\mathit{udistr}}}
\newcommand{\phiforbidwht}{\varphi^{\textit{forbid}}_{\wht \cdot \wht}}
\newcommand{\phiforbidshdw}{\varphi^{\textit{forbid}}_{\shdw \cdot \shdw}}
\newcommand{\phibase}{\varphi_{\textit{base}}^{\textit{lem}\ref{lemma:fo2-shadowy}}}
\newcommand{\psishadowyfo}{\psi_{\textit{shadowy}}^{\textit{FO}}}
\newcommand{\fottr}[2]{\mathfrak{tr}_{#1}(#2)}

\knowledge{shadowy}[shadowness]{notion, color=black}
\knowledge{truly~$\Sigma$-shadowy}[truly-shadowness]{notion, color=black}
\knowledge{strongly shadowy}[strongly shadowness]{notion, color=black}
\knowledge{importunate}{notion, color=black}
\knowledge{dehalfication}{notion, color=black}

\knowledge{psishadowy}[\psishadowy]{notion, color=black}
\knowledge{psitrulySigmashadowy}[$\psitrulySigmashadowy$]{notion, color=black}
\knowledge{psitrulySigmaMinsky}[$\psitrulySigmaMinsky$ | truly~$\SigmaMinsky$-shadowy]{notion, color=black}
\knowledge{phitransfer}[$\phitransfer$ | $\phitransfer {\sigma }{\sigmashdw }$]{notion, color=black}
\knowledge{phiodd}[$\phiodd$]{notion, color=black}
\knowledge{phiisequal}[$\phiisequal{\alpha}{\beta}$]{notion, color=black}
\knowledge{proj}[$\proj{\Sigma}{\aword}$ | $\proj {\{\wht ,\shdw \}}{\aword }$]{notion, color=black}
\knowledge{countingterm}[counting terms]{notion, color=black}

\knowledge{minsky}[Minsky machine | $\minsky$ | Two-counter Machines | deterministic Minsky machine | Minsky machines | Minsky Machines]{notion, color=black}

\title{``Most of'' leads to undecidability:\\ Failure of adding frequencies to LTL}

\author{Bartosz Bednarczyk$^{1,2}$ \orcid{0000-0002-8267-7554}{} \; and \; Jakub Michaliszyn$^{2}$ \orcid{0000-0002-5053-0347}}
\date{$^1$ Computational Logic Group, Technische Universit{\"a}t Dresden, Germany\\
$^2$ Institute of Computer Science, University of Wroc{\l}aw, Poland}

\begin{document}
\maketitle

%%%%%%%%%%%%%%%%%%%%%%%%%%%%%%%%%%%%%%%%% ABSTRACT %%%%%%%%%%%%%%%%%%%%%%%%%%%%%%%%%%%%%%%%%%%%%%%%%%%%%%%%%%%%%%%%%%%%

\begin{abstract}
Linear Temporal Logic (LTL) interpreted on finite traces is a robust specification framework popular in formal verification. 
However, despite the high interest in the logic in recent years, the topic of their quantitative extensions is not yet fully explored. 
The main goal of this work is to study the effect of adding weak forms of percentage constraints (\eg that \emph{most of} the positions in the past satisfy a given condition, or that $\sigma$ is the \emph{most-frequent} letter occurring in the past) to fragments of LTL. 
Such extensions could potentially be used for the verification of influence networks or statistical reasoning.
Unfortunately, as we prove in the paper, it turns out that percentage extensions of even tiny fragments of LTL have undecidable satisfiability and model-checking problems. 
Our undecidability proofs not only sharpen most of the undecidability results on logics with arithmetics interpreted on words known from the literature, but also are fairly simple.

We also show that the undecidability can be avoided by restricting the allowed usage of the negation, and briefly discuss how the undecidability results transfer to first-order logic on words.
\end{abstract}

%%%%%%%%%%%%%%%%%%%%%%%%%%%%%%%%%%%%%%%%% INTRODUCTION %%%%%%%%%%%%%%%%%%%%%%%%%%%%%%%%%%%%%%%%%%%%%%%%%%%%%%%%%%%%%%%%%%%%

\section{Introduction}\label{sec:intro}
Linear Temporal Logic~\cite{Pnueli77} (LTL) interpreted on finite traces is a robust logical framework used in formal verification~\cite{BaierM06,GiacomoV13,GiacomoV15}. 
However, LTL is not perfect: it can express whether some event happens or not, but it cannot provide any insight on how frequently such 
an event occurs or for how long such an event took place. 
In many practical applications, such \emph{quantitative} information is important: think of optimising a server based on how frequently it receives 
messages or optimising energy consumption knowing for how long a system is usually used in rush hours. 
Nevertheless, there is a solution: one can achieve such goals by adding quantitative features to LTL.

It is known that adding quantitative operators to $\LTL$ often leads to undecidability. 
The proofs, however, typically involve operators such as ``next'' or ``until'', and are often quite complicated (see the discussion on the related work below). 
In this work, we study the logic $\LTLF$, a fragment of $\LTL$ where the only allowed temporal operator is ``sometimes in the future'' $\F{}$. 
We extend its language with two types of operators, sharing a similar ``percentage'' flavour: with the \emph{Past-Majority} $\MajP{\varphi}$ operator
(stating that most of the past positions satisfy a formula $\varphi$), and with the \emph{Most-Frequent-Letter} $\MostFreq{\sigma}$ predicates (meaning that the letter $\sigma$ is among the most frequent letters appearing in the past).
These operators can be used to express a number of interesting properties, such as \emph{if a process failed to enter the critical section, then the other process was in the critical section the majority of time}.
Of course, for practical applications, we could also consider richer languages, such as parametrised versions of these operators, \eg stating that \emph{at least a fraction $p$ of positions in the past satisfies a formula}. 
However, we show, as our main result, that even these very simple percentage operators raise undecidability when combined with $\F{}$.

To make the undecidability proof for both operators similar, we define an intermediate operator, $\Half{}$, which is satisfied when exactly half of the past positions satisfy a given formula.
The $\Half{}$ operator can be expressed easily with $\MajP{}$, but not with $\MostFreq{}$ --- we show, however, that we can simulate it to an extent enough to show the undecidability.
Our proof method relies on enforcing a model to be in the language $(\{\wht\}\{\shdw\})^+$, for some letters $\wht$ and $\shdw$, which a priori seems to be impossible without the ``next'' operator. 
Then, thanks to the specific shape of the models, we show that one can ``transfer'' the truth of certain formulae from positions into their successors, hence the ``next'' operator can be partially expressed.
With a combination of these two ideas, we show that it is possible to write equicardinality statements in the logic. 
Finally, we perform a reduction from the reachability problem of \kl{Two-counter Machines}~\cite{Minsky}. 
In the reduction, the equicardinality statements will be responsible for handling zero-tests. 
The idea of transferring predicates from each position into its successor will be used for switching the machine into its next configuration. 

The presented undecidability proof of $\LTL$ with percentage operators can be adjusted to extensions of fragments of first-order logic on finite words. 
We show that $\FOtMaj[<]$, \ie{} the two-variable fragment of first-order logic admitting the majority quantifier $\Maj{}$ and linear order predicate~$<$ has an undecidable satisfiability problem.
Here the meaning of a formula~$\Maj{x} . \varphi(x, y)$ is that at least a half of possible interpretations of $x$ satisfies $\varphi(x,y)$. 
Our result sharpens an existing undecidability proof for (full) $\FO$ with Majority from~\cite{Lange04} (since in our case the number of variables is limited) but also $\FOt[<, \succ]$ with arithmetics from~\cite{LodayaS17} (since our counting mechanism is weaker and the successor relation $\succ$ is disallowed).

On the positive side, we show that the undecidability heavily depends on the presence of the negation in front of the percentage operators. 
To do so, we introduce a logic, extending the full $\LTL$, in which the usage of percentage operators is possible, but suitably restricted. 
For this logic, we show that the satisfiability problem is decidable.

All the above-mentioned results can be easily extended to the model checking problem, where the question is whether a given Kripke structure satisfies a given formula.

\subsection{Related work}\label{subsec:related-work}
The first paper studying the addition of quantitative features to logic was~\cite{KlaedtkeR03}, where the authors proved undecidability of 
Weak MSO with Cardinalities. 
They also developed a model of so-called Parikh Automaton, a finite automaton imposing a semi-linear constraint on the set of its final configurations. 
Such an automaton was successfully used to decide logics with counting as well as logics on data words~\cite{Niewerth16,FigueiraL15}.
Its expressiveness was studied in~\cite{CadilhacFM12}. 

Another idea in the realm of quantitative features is availability languages~\cite{HoenickeMO10}, which extend regular expressions by numerical occurrence constraints on the letters. 
However, their high expressivity leads to undecidable emptiness problems.
Weak forms of arithmetics have also attracted interest from researchers working on temporal logics. 
Several extensions of LTL were studied, including extensions with counting~\cite{LaroussinieMP10}, periodicity constraints~\cite{Demri06}, accumulative values~\cite{BokerCHK14}, discounting~\cite{AlmagorBK14}, averaging~\cite{BouyerMM14} and frequency constraints~\cite{BolligDL12}.
A lot of work was done to understand LTL with timed constraints, \eg{} a metric LTL was considered in~\cite{OuaknineW07}. 
However, its complexity is high and its extensions are undecidable~\cite{AlurH93}. 

Arithmetical constraints can also be added to the First-Order logic~$(\FO)$ on words via so-called counting quantifiers. 
It is known that weak MSO on words is decidable with threshold counting and modulo-counting (thanks to the famous B{\"u}chi theorem~\cite{Buchi60}), while even $\FO$ on words with percentage quantifiers becomes undecidable~\cite{Lange04}. 
Extensions of fragments of $\FO$ on words are often decidable, \eg the two-variable fragment $\FOt$ with counting~\cite{Charatonik016a} or $\FOt$ with modulo-counting~\cite{LodayaS17}.  
The investigation of decidable extensions of $\FOt$ is limited by the undecidability of $\FOt$ on words with Presburger constraints~\cite{LodayaS17}. 

Among the above-mentioned logics, the formalisms of this paper are most similar to Frequency LTL~\cite{BolligDL12}. 
The satisfiability problem for Frequency LTL was claimed to be undecidable, but the undecidability proof as presented in~\cite{BolligDL12} is bugged (see~\cite[Sec. 8]{BouyerMM14} for discussion).
It was mentioned in~\cite{BouyerMM14} that the undecidability proof from~\cite{BolligDL12} can be patched, but no correction was published so far. 
Our paper not only provides a valid proof but also sharpens the result, as we use a way less expressive language (\eg{} we are allowed to use neither the  ``until'' operator nor the  ``next'' operator). 
We also believe that our proof is simpler.
The second-closest formalism to ours is average-LTL~\cite{BouyerMM14}. 
The main difference is that the averages of average-LTL are computed based on the future, while in our paper, the averages are based on the past. 
The second difference, as in the previous case, is that their undecidability proof uses more expressive operators, such as the  ``until'' operator.

%%%%%%%%%%%%%%%%%%%%%%%%%%%%%%%%%%%%%%%%% PRELIMINARIES %%%%%%%%%%%%%%%%%%%%%%%%%%%%%%%%%%%%%%%%%%%%%%%%%%%%%%%%%%%%%%%%%%%%

\section{Preliminaries} \label{sec:preliminaries}
We recall classical definitions concerning logics on words and temporal logics (\cf~\cite{Demri2016}).

\subsection{Words and logics}
Let $\APset$ be a countably-infinite set of \emph{atomic propositions}, called here also \emph{letters}. 
A finite \emph{word} $\aword \in (2^{\APset})^{*}$ is a non-empty finite sequence of \emph{positions} labelled with sets of letters from $\APset$. 
A~set of words is called a \emph{language}. 
Given a word $\aword$, we denote its $i$-th position with $\aword_i$ (where the first position is $\aword_0$) and its prefix up to the $i$-th position with $\aword_{\leq i}$.
We employ the letters $i,j,p,q$ to denote positions.
With $|\aword|$ we denote the length of $\aword$.

The syntax of $\LTLF$, a fragment of $\LTL{}$ with only the \emph{finally} operator $\F{}$, is defined as usual with the~grammar:
\[ \varphi, \varphi' ::= a \; (\text{with} \; a \in \APset) \; \mid \; \neg \varphi \; \mid \; \varphi \wedge \varphi' \; \mid \; \F{\varphi}. \]

The satisfaction relation $\models$ is defined for words as follows:
\begin{center}
\begin{tabular}[t]{lll}
$\aword,i \models a$ & \text{if} & $a \in \aword_i$ \\
$\aword,i \models \neg \varphi$ & \text{if} & \text{not} \; $\aword,i \models \varphi$ \\
$\aword,i \models \varphi_1 \wedge \varphi_2$ & \text{if} & $\aword,i \models \varphi_1$ \; \text{and} \; $\aword,i \models \varphi_2$  \\
$\aword,i \models \F{\varphi}$ & \text{if} & $\exists{j} \; \text{such that} \; \; |\aword| > j \geq i \; \text{and} \; \aword,j \models \varphi$.
\end{tabular}
\end{center}

We write $\aword \models \varphi$ if $\aword, 0 \models \varphi$. 
The usual Boolean connectives: $\top, \bot, \vee, \rightarrow, \leftrightarrow$ can be defined, hence we will use them as abbreviations. 
Additionally, we use the \emph{globally} operator $\G{\varphi} := \neg \F{\neg \varphi}$ to speak about events happening globally in the future.

\subsection{Percentage extension}
In our investigation, \emph{percentage operators} $\MajPname{}$, $\MostFreqname{}$ and $\Halfname$ are added to $\LTLF$. 

The operator $\MajP{} \varphi$ (read as: \emph{majority in the past}) is satisfied if at least half of the positions in the past satisfy $\varphi$:
\begin{center}
\begin{tabular}[t]{lll}
$\aword,i \models \MajP{\varphi}$ & \text{if} & $|\setof{j < i \colon {\aword, j} \models {\varphi}}| \geq \frac{i}{2}$
\end{tabular}
\end{center}

For example, the formula $\G{(r \leftrightarrow \neg g)} \land \G{\MajP{r}} \land \G{\F{\left(g \land \MajP{g} \right)}}$ is true over words where each \emph{request} $r$ is eventually fulfilled by a \emph{grant} $g$, and where each grant corresponds to at least one request. 
This can be also seen as the language of balanced parentheses, showing that with the operator $\MajPname$ one can define properties that are not regular.

The operator $\MostFreq{} \sigma$ (read as: \emph{most-frequent letter in the past}), for $\sigma \in \APset$, is satisfied if~$\sigma$ is among the letters with the highest number of appearances in the past, \ie{}
\begin{center}
\begin{tabular}[t]{lll}
$\aword,i \models \MostFreq{}{\sigma}$ & \text{if} & $\forall \tau \in \APset.$
$|\setof{j < i \colon {\aword, j} \models {\sigma}}| \geq |\setof{j < i \colon {\aword, j} \models {\tau}}|$
\end{tabular}
\end{center}

For example, the formula $\G{\neg (r \land g)} \land \G{~\MostFreq{r}} \land \G{\F{\left(g \land \MostFreq{g}\right)}}$ again defines words where each request is eventually fulfilled, but this time the formula allows for states where nothing happens (\ie when both $r$ and $g$ are false).

The last operator, $\Halfname$ is used to simplify the forthcoming undecidability proofs. 
This operator can be satisfied only at even positions, and its intended meaning is \emph{exactly half of the past positions satisfy a given~formula}. 
\begin{center}
\begin{tabular}[t]{lll}
$\aword,i \models \Half{\varphi}$ & \text{if} & $|\setof{j < i \colon {\aword, j} \models {\varphi}}| = \frac{i}{2}$
\end{tabular}
\end{center}
It is not difficult to see that the operator $\Half{\varphi}$ can be defined in terms of the past-majority operator as~$\MajP(\varphi) \wedge \MajP(\neg \varphi)$ and that $\Half{\varphi}$ can be satisfied only at even positions.

In the next sections, we distinguish different logics by enumerating the allowed operators in the subscripts, \eg $\LTLFMajP$ or $\LTLFMostFreq$.

\subsection{Computational problems}
\emph{Kripke structures} are commonly used in verification to formalise abstract models. 
A Kripke structure is composed of a finite set $S$ of \emph{states}, a set of \emph{initial} states $I \subseteq S$, a total \emph{transition} relation $R \subseteq S \times S$, and a finite \emph{labelling function} $\ell : S \rightarrow 2^{\APset}$. 
A \emph{trace} of a Kripke structure is a finite word $\ell(s_0), \ell(s_1), \ldots, \ell(s_k)$  for any $s_0, s_1, \dots, s_k$ satisfying $s_0 \in I$ and $(s_{i},s_{i+1}) \in R$ for all $i < k$. 

The \emph{model-checking problem} amounts to checking whether \emph{some} trace of a given Kripke structure satisfies a given formula $\varphi$.
In the \emph{satisfiability problem}, or simply in \emph{SAT}, we check whether an input formula $\varphi$ has a \emph{model}, \ie a finite word $\aword$ witnessing $\aword \models \varphi$.

%%%%%%%%%%%%%%%%%%%%%%%%%%%%%%%%%%%%%%%%% PLAYING WITH HALF %%%%%%%%%%%%%%%%%%%%%%%%%%%%%%%%%%%%%%%%%%%%%%%%%%%%%%%%%%%%%%%%%%%%

\section{Playing with Half Operator} \label{subsec:playingwithmaj}
Before we jump into the encoding of Minsky machines, we present some exercises to help the reader understand the expressive power of the logic $\LTLFHalf$. 
The tools established in the exercises play a vital role in the undecidability proofs provided in the following section.

We start from the definition of \kl{shadowy} words.

\begin{definition} \label{def:shadowy}
Let $\wht$ and $\shdw$ be fixed distinct atomic propositions from $\APset$. 
A word~$\aword$ is \newnotion{shadowy} if its length is even, all even positions of $\aword$  are labelled with $\wht$, all odd positions of $\aword$ are labelled with $\shdw$, and no position is labelled with both letters.
\end{definition}
%\\
\begin{center}
\scalebox{1}{
    \begin{tikzpicture}[minimum size=11mm]
      \begin{scope}[ptrond]
      \draw (0,0) node[blanc] (A1) {$\wht$};
      \draw (2,0) node[gris] (A2) {$\shdw$};
      \draw (4,0) node[blanc] (A3) {$\wht$};
      \draw (6,0) node[gris] (A4) {$\shdw$};
      \draw (8,0) node[blanc] (A5) {$\wht$};
      \draw (10,0) node[gris] (A6) {$\shdw$};
      \end{scope}
      \foreach \i/\j in {1/2, 3/4, 4/5, 5/6}
        {\draw[-latex'] (A\i) -- (A\j);}
      \foreach \i/\j in {2/3, 5/6}
        {\draw[dashed] (A\i) -- (A\j);}
    \end{tikzpicture}
}
\end{center}
We will call the positions satisfying~$\wht$ simply \emph{white} and their successors satisfying~$\shdw$ simply their \emph{shadows}.

The following exercise is simple in $\LTL{}$, but becomes much more challenging without the~$\X{}$ operator.

\begin{exercise} \label{ex:1}
There is an $\LTLFHalf$ formula $\newnotion{\psishadowy}$ defining \kl{shadowy} words.
\end{exercise}
\begin{proof}[Solution]
We start with the ``base'' formula~$\phiinit \deffnoindent \wht \land \G (\wht \leftrightarrow \neg\shdw) \wedge \G(\wht \rightarrow \F{\shdw})$, which states that the position $0$ is labelled with $\wht$, each position is labelled with exactly one letter among $\wht, \shdw$ and that every white eventually sees a shadow in the future. 
What remains to be done is to ensure that only odd positions are shadows and that only even positions are white.

In order to do that, we employ the formula \newnotion{$\phiodd$} $\deffnoindent \G((\Half{\wht}) \leftrightarrow \wht)$. 
Since $\Halfname{}$ is never satisfied at odd positions, the formula $\phiodd$ stipulates that odd positions are labelled with $\shdw$. 
An inductive argument shows that all the even positions are labelled with $\wht$: for the position $0$, it follows from $\phiinit$. 
For an even position $p>0$, assuming (inductively) that all even positions are labelled with $\wht$, the formula $\phiodd$ ensures that $p$ is labelled with $\wht$.

Putting it all together, the formula $\psishadowy \deffnoindent \phiinit \land \phiodd$ is as required.
\end{proof}

In the next exercise, we show that it is possible to transfer the presence of certain letters from white positions into their shadows. 
It justifies the usage of ``shadows'' in the paper.

We introduce the so-called \newnotion{counting terms}.
For a formula $\varphi$, word $\str{w}$ and a position~$p$, by~$\mysharp{<}{\varphi}{\aword, p}$ we denote the total number of positions among $0, \ldots, p{-}1$  satisfying~$\varphi$, \ie the size of $\{ p' < p \mid \aword, p' \models \varphi \}$.
We omit $\aword$ in counting terms if it is known from the context. 

\begin{exercise} \label{ex:2}
Let $\sigma$ and $\sigmashdw$ be distinct letters from $\APset \setminus \{ \wht, \shdw \}$. 
There is an $\LTLFHalf$ formula~\newnotion{$\phitransfer{\sigma}{\sigmashdw}$}, such that $\aword \models \phitransfer{\sigma}{\sigmashdw}$ iff: 
\begin{enumerate}
    \item\label{item:1:ex:2} $\aword$ is \kl{shadowy},
    \item\label{item:2:ex:2} only white (resp., shadow) positions of~$\aword$ can be labelled $\sigma$ (resp., $\sigmashdw$) and
    \item\label{item:3:ex:2} for any even position $p$ we have: $\aword,p \models \sigma \Leftrightarrow \aword,p{+}1 \models \sigmashdw$.
\end{enumerate}
\end{exercise}
\begin{center}
    \begin{tikzpicture}[minimum size=11mm]
      \begin{scope}[ptrond]
      \draw (0,0) node[blanc] (A1) {$\wht$};
      \draw (2,0) node[gris] (A2){$\shdw$};
      \draw (4,0) node[blanc] (A3) {$\wht$};
      \draw (6,0) node[gris] (A4) {$\shdw$};
      \draw (8,0) node[blanc] (A5) {$\wht$};
      \draw (10,0) node[gris] (A6){$\shdw$};
      \draw (0,-0.33) node[] (B1) {} node {$\sigma$};
      \draw (2,-0.315) node[] (B2) {} node {$\sigmashdw$};
      \draw (4,-0.33) node[] (B3) {} node {$\neg\sigma$};
      \draw (6,-0.315) node[] (B4) {} node {$\neg\sigmashdw$};
      \draw (8,-0.33) node[] (B3) {} node {$\neg\sigma$};
      \draw (10,-0.315) node[] (B4) {} node {$\neg\sigmashdw$};

      \end{scope}
      \foreach \i/\j in {1/2, 3/4, 4/5, 5/6}
        {\draw[-latex'] (A\i) -- (A\j);}
      \foreach \i/\j in {2/3, 5/6}
        {\draw[dashed] (A\i) -- (A\j);}
    \end{tikzpicture}
\end{center}
\begin{proof}[Solution]
Note that the first two conditions can be expressed with the conjunction of $\psishadowy$, $\G{(\sigma \rightarrow \wht)}$ and $\G{(\sigmashdw \rightarrow \shdw)}$.
The last condition is more involving. 
Assuming that the words under consideration satisfy conditions~\ref{item:1:ex:2}--\ref{item:2:ex:2}, it is easy to see that the third condition is equivalent to expressing that all white positions $p$ satisfy the equation~$(\heartsuit)$:
\[ (\heartsuit): \; \; \mysharp{<}{\wht \wedge \sigma}{\aword, p} = \mysharp{<}{\shdw \wedge \sigmashdw}{\aword, p} \]
supplemented with the condition $(\diamondsuit)$, ensuring that the last white position satisfies the condition~\ref{item:3:ex:2}, \ie{}
\[ (\diamondsuit): \; \; \text{for the last white position} \; p \; \text{we have:} \; \aword,p \models \sigma \Leftrightarrow \aword,p{+}1 \models \sigmashdw. \]

For a curious reader we present the proof of this claim below.
\begin{claim}
Let $\aword$ be a word satisfying the conditions~\ref{item:1:ex:2}--\ref{item:2:ex:2}.
Then $\aword$ satisfies the condition~\ref{item:3:ex:2} iff $\aword$ satisfies $(\diamondsuit)$ and for all white positions $p$ the equation $(\heartsuit)$ holds.  
\end{claim}
\begin{proof}
Assume that a word $\aword$ satisfies the conditions~\ref{item:1:ex:2}--\ref{item:3:ex:2}.
Then the condition $(\diamondsuit)$ follows immediately from the condition~\ref{item:3:ex:2}.
To see that for all white $p$ the equation $(\heartsuit)$ holds, we employ induction over white positions in $\aword$.
In the base case we have $p=0$ (due to the shadowness). 
Since there are no positions before $p$, we conclude that both the LHS and the RHS of $(\heartsuit)$ are equal to $0$, thus $(\heartsuit)$ holds.
Now, take any white position $p$ and assume that for all white $p' < p$ satisfy~$(\heartsuit)$.
Note that due to shadowness of $\aword$ and the fact that $p$ is white, the LHS of~$(\heartsuit)$ is equal~to:
\[
    \mysharp{<}{\wht \wedge \sigma}{\aword, p} = \mysharp{<}{\wht \wedge \sigma}{\aword, p{-}2} + \textit{is-labelled-with}^{\sigma}_{p{-}2},
\]
where $\textit{is-labelled-with}^{\sigma}_{p{-}2}$ is equal to $1$ if $\aword, p{-}2 \models \sigma$ and $0$ otherwise.
Analogously, the RHS of $(\heartsuit)$ is equal~to:
\[
    \mysharp{<}{\shdw \wedge \sigmashdw}{\aword, p} = \mysharp{<}{\shdw \wedge \sigmashdw}{\aword, p{-}2} + \textit{is-labelled-with}^{\sigmashdw}_{p-1}.
\]
From the inductive assumption, we infer $\mysharp{<}{\wht \wedge \sigma}{\aword, p{-}2} = \mysharp{<}{\shdw \wedge \sigmashdw}{\aword, p{-}2} $.
Moreover, by applying the condition~\ref{item:3:ex:2} to the position $p{-}2$, we get $\textit{is-labelled-with}^{\sigma}_{p{-}2} = \textit{is-labelled-with}^{\sigmashdw}_{p-1}$.
Hence, $(\heartsuit)$ holds for all whites.

For the opposite direction, assume that a word $\aword$ satisfies the conditions~\ref{item:1:ex:2}--\ref{item:2:ex:2} as well as $(\heartsuit)$ and $(\diamondsuit)$.
Ad absurdum, assume that $\aword$ does not satisfy the condition~\ref{item:3:ex:2} and let $p$ be the smallest white position violating the condition~\ref{item:3:ex:2}.
If $|\aword| = 2$ or $p$ is the last position of $\aword$, then we have contradiction with $(\diamondsuit)$. 
Thus, $|\aword| > 2$ and~$p$ is not the last white position in $\aword$.
Moreover, $p\neq0$.
Indeed, if $p=0$ then we have contradiction with $(\heartsuit)$ applied to $p{=}2$ since one side of $(\heartsuit)$ is equal to $1$, while the other is equal to $0$. %[and they are supposed to be equal].
Hence, $p$ is neither the first white position nor the last one.
From $(\heartsuit)$ applied to $p{+}2$ we get the equality~$\mysharp{<}{\wht \wedge \sigma}{\aword, p{+}2} = \mysharp{<}{\shdw \wedge \sigmashdw}{\aword, p{+}2}$.
Additionally, from $(\heartsuit)$ applied to $p$, we obtain $\mysharp{<}{\wht \wedge \sigma}{\aword, p} = \mysharp{<}{\shdw \wedge \sigmashdw}{\aword, p}$.
Reasoning similarly to the first part of the proof, we know that:
\[
    \mysharp{<}{\wht \wedge \sigma}{\aword, p{+}2} = \mysharp{<}{\wht \wedge \sigma}{\aword, p} + \textit{is-labelled-with}^{\sigma}_{p},
\]
and
\[
    \mysharp{<}{\shdw \wedge \sigmashdw}{\aword, p{+}2} = \mysharp{<}{\shdw \wedge \sigmashdw}{\aword, p} + \textit{is-labelled-with}^{\sigmashdw}_{p+1}
\]
hold, which clearly implies the equality $\textit{is-labelled-with}^{\sigma}_{p} = \textit{is-labelled-with}^{\sigmashdw}_{p{+}1}$.
But such equality does not hold due to the fact that $p$ violates condition~\ref{item:3:ex:2}. 
A contradiction.
Thus $\aword$ satisfies condition~\ref{item:3:ex:2}.
\end{proof}

Going back to Exercise~\ref{ex:2}, we show how to define $(\heartsuit)$ and $(\diamondsuit)$ in $\LTLFHalf$, taking advantage of \kl{shadowness} of the intended models.
Take an arbitrary white position $p$ of $\aword$. The equation $(\heartsuit)$ for $p$ is clearly equivalent to:
\[
(\heartsuit'): \; \; \mysharp{<}{\wht \wedge \sigma}{\aword, p} +
\left(\frac{p}{2} - \mysharp{<}{\shdw \wedge \sigmashdw}{\aword, p} \right) = \frac{p}{2}
\]
Since $p$ is even, we infer that $\frac{p}{2} \in \N$. 
From the \kl{shadowness} of $\aword$, we know that there are exactly $\frac{p}{2}$ shadows in the past of $p$. 
Moreover, each shadow satisfies either $\sigmashdw$ or $\neg\sigmashdw$.  
Hence, the expression $\frac{p}{2} - \mysharp{<}{\shdw \wedge \sigmashdw}{\aword, p}$ from $(\heartsuit')$, can be replaced with $\mysharp{<}{\shdw \wedge \neg\sigmashdw}{\aword, p}$. 
Finally, since $\wht$ and $\shdw$ label disjoint positions, the property that every white position $p$ satisfies $(\heartsuit)$ can be written as an $\LTLFHalf$ formula 
  $
  \varphi_{(\heartsuit)} \deffnoindent 
  \G{\left( 
  \wht \rightarrow \Half{( [\wht \land \sigma] 
  \vee [\shdw \land \neg \sigmashdw])}
  \right)}
  $.
Its correctness follows from the correctness of each arithmetic transformation and the semantics of $\LTLFHalf{}$.  

% \noindent 
For the property $(\diamondsuit)$, we first need to define formulae detecting the last and the second to last positions of the model. 
Detecting the last position is easy: since the last position of $\aword$ is shadow, it is sufficient to express that it sees only shadows in its future, \ie $\philst \deffnoindent \G(\shdw)$.
Similarly, a position is second to last if it is white and it sees only white or last positions in the future, which results in a formula $\phistl \deffnoindent \wht \land \G(\wht \vee \philst)$.
Note that the correctness of $ \philst$ and $\phistl$ follows immediately from shadowness.
Hence, we can define the formula~$\varphi_{(\diamondsuit)}$  as~$\F(\phistl \land \sigma) \leftrightarrow \F(\philst \land \sigmashdw)$.
The conjunction of $\varphi_{(\heartsuit)}$ and $\varphi_{(\diamondsuit)}$ formulae gives us to~$\phitransfer{\sigma}{\sigmashdw}$.
\end{proof}

We consider a generalisation of \kl{shadowy} models, where each shadow mimics all letters from a finite set $\Sigma \subseteq \APset$ rather than just a single letter $\sigma$. 
Such a generalisation is described below. 
In what follows, we always assume that for each $\sigma \in \Sigma$ there is a unique $\sigmashdw$, which is different from $\sigma$, and $\sigmashdw \not\in \Sigma$. 
Moreover, we always assume that $\sigma_1 \neq \sigma_2$ implies $\sigmashdw_1 \neq \sigmashdw_2$.
\begin{definition}
Let $\Sigma \subseteq \APset \setminus \{ \wht, \shdw \}$ be a finite set.
A \kl{shadowy} word $\aword$ is called \newnotion{truly~$\Sigma$-shadowy}, if for every letter $\sigma \in \Sigma$ only the white (resp. shadow) positions of $\aword$ can be labelled with $\sigma$ (resp. $\sigmashdw$) and every white position $p$ of $\aword$ satisfies $\aword,p \models \sigma \Leftrightarrow \aword,p{+}1 \models \sigmashdw$.
\end{definition}
\begin{center}
\scalebox{1}{
    \begin{tikzpicture}[minimum size=11mm]
      \begin{scope}[ptrond]
      \draw (0,0) node[blanc] (A1){$\wht$};
      \draw (2,0) node[gris] (A2){$\shdw$};
      \draw (4,0) node[blanc] (A3) {$\wht$};
      \draw (6,0) node[gris] (A4) {$\shdw$};
      \draw (8,0) node[blanc] (A5) {$\wht$};
      \draw (10,0) node[gris] (A6) {$\shdw$};

      \draw (0,-0.33) node[] (B1) {} node {\scriptsize $\alpha, \beta$};
      \draw (2,-0.315) node[] (B2) {} node {\scriptsize$\alphashdw, \betashdw$};
      \draw (4,-0.33) node[] (B3) {} node {\scriptsize$\neg\alpha, \beta$};
      \draw (6,-0.315) node[] (B4) {} node {\scriptsize$\neg\alphashdw, \betashdw$};
      \draw (8,-0.33) node[] (B3) {} node {\scriptsize$\alpha, {\neg}\beta$};
      \draw (10,-0.315) node[] (B4) {} node {\scriptsize$\alphashdw, {\neg}\betashdw$};

      \end{scope}
      \foreach \i/\j in {1/2, 3/4, 4/5, 5/6}
        {\draw[-latex'] (A\i) -- (A\j);}

      \foreach \i/\j in {2/3}
        {\draw[dashed] (A\i) -- (A\j);} 
    \end{tikzpicture}
}  
\end{center}
Knowing the solution for the previous exercise, it is easy to come up with a formula \newnotion{$\psitrulySigmashadowy$} defining \kl{truly~$\Sigma$-shadowy} models: just take the conjunction of $\psishadowy$ and $\phitransfer{\sigma}{\sigmashdw}$ over all letters $\sigma \in \Sigma$.
The correctness follows immediately from from Exercise~\ref{ex:2}.
\begin{corollary}\label{col:trulyshadowy}
The formula $\psitrulySigmashadowy$ defines the language of \kl{truly~$\Sigma$-shadowy} words.
\end{corollary}

The next exercise shows how to compare cardinalities in $\LTLFHalf$ over 
\kl{truly~$\Sigma$-shadowy} models. We are not going to introduce any novel 
techniques here, but the exercise is of great importance: it is used in the next
section to encode zero tests of Minsky machines.

\begin{exercise} \label{ex:3}
Let $\Sigma$ be a finite subset of $\APset \setminus \{ \wht, \shdw \}$  and let $\alpha {\neq} \beta \in \Sigma$. 
There exists an~$\LTLFHalf$ formula \newnotion{$\phiisequal{\alpha}{\beta}$} such that for any \kl{truly~$\Sigma$-shadowy} word $\aword$ and any of its white positions $p$: the equivalence $\aword, p \models \phiisequal{\alpha}{\beta} \Leftrightarrow \mysharp{<}{\wht \wedge \alpha}{\aword, p} = \mysharp{<}{\wht \wedge \beta}{\aword, p}$ holds.
\end{exercise}
\begin{center}
\scalebox{1}{
    \begin{tikzpicture}[minimum size=11mm]
      \begin{scope}[ptrond]
      \draw (0,0) node[blanc] (A1) {$\wht$};
      \draw (2,0) node[gris] (A2) {$\shdw$};
      \draw (4,0) node[blanc] (A3) {$\wht$};
      \draw (6,0) node[gris] (A4){$\shdw$};
      \draw (8,0.0) node[blanc] (A5) {} node {};
      \draw (10,0) node[] (A6) {};

      \draw (8, 1.1) node[] (B1) {$\phiisequal{\alpha}{\beta}$};

      \draw (0,-0.33) node[] (B1) {} node {\scriptsize$\alpha, \neg\beta$};
      \draw (2,-0.315) node[] (B2) {} node {\scriptsize$\alphashdw, \neg\betashdw$};
      \draw (4,-0.33) node[] (B3) {} node {\scriptsize$\neg\alpha, \beta$};
      \draw (6,-0.315) node[] (B4) {} node {\scriptsize$\neg\alphashdw, \betashdw$};

      \draw (3,1.36) node[] (C) {$\# \alpha = \# \beta $};

      \end{scope}
      \foreach \i/\j in {1/2, 3/4, 4/5}
        {\draw[-latex'] (A\i) -- (A\j);}

      \foreach \i/\j in {2/3, 5/6}
        {\draw[dashed] (A\i) -- (A\j);} 

    \draw[decoration={calligraphic brace,amplitude=10pt}, 
    decorate, line width=0.3pt] (-0.4,0.81) node {} -- (6.4, 0.81);

    \end{tikzpicture}
}  
\end{center}
%
% We left this exercise to the reader as is it an easy modification of the previous exercise. A hint and a solution are provided below.
% \begin{proof}[Hint]
% Follow the previous exercise. 
% The main difficulty is to express the equality of counting terms, written as $\LHS = \RHS$. 
% Note that it is clearly equivalent to $\LHS + (\frac{p}{2} - \RHS) = \frac{p}{2}$. Unfold $\frac{p}{2}$ on the left hand side, \ie replace it with the total number of shadows in the past.
% Use the fact that $\aword$ satisfies $\phitransfer{\sigma}{\sigmashdw}$, which implies the equality $\mysharp{<}{\wht \wedge \beta}{\aword, p} = \mysharp{<}{\shdw \wedge \betashdw}{\aword, p}$.
% Finally, get rid of subtraction and write an $\LTLFHalf$ formula by employing $\Halfname{}$.
% % \def\rotatecharone#1{\rotatebox[origin=c]{180}{#1}}
% % \textbf{Solution.} \rotatecharone{The formula is $\phiisequal{\alpha}{\beta} \deffnoindent \Half ( [\wht \land \alpha] \vee [\shdw \land \neg \betashdw] )$.}
% \end{proof}
\begin{proof}
We proceed similarly to Exercise~\ref{ex:2}, but actually the forthcoming proof is easier.
Let us fix a white position $p$ from $\aword$. 
We would like to express that $\mysharp{<}{\wht \wedge \alpha}{\aword, p} = \mysharp{<}{\wht \wedge \beta}{\aword, p}$ holds, which is equivalent to 
expressing $\mysharp{<}{\wht \wedge \alpha}{\aword, p} - \mysharp{<}{\wht \wedge \beta}{\aword, p}= 0$.
Since $p$ is white, then $\frac{p}{2} \in \N$, so we can add $\frac{p}{2}$ to both sides.
Moreover, $\frac{p}{2}$ is equal to the total number of shadows in the past of $p$, hence our initial equation is equivalent to:
\[\mysharp{<}{\wht \wedge \alpha}{\aword, p} - \mysharp{<}{\wht \wedge \beta}{\aword, p} = 0\]
\[\mysharp{<}{\wht \wedge \alpha}{\aword, p} - \mysharp{<}{\wht \wedge \beta}{\aword, p} + \frac{p}{2} = \frac{p}{2}\]
\[\mysharp{<}{\wht \wedge \alpha}{\aword, p} - \mysharp{<}{\wht \wedge \beta}{\aword, p} + \mysharp{<}{\shdw}{\aword, p} = \frac{p}{2}\]
Since $\aword$ satisfies $\phitransfer{\sigma}{\sigmashdw}$, we know that the equality $\mysharp{<}{\wht \wedge \beta}{\aword, p} = \mysharp{<}{\shdw \wedge \betashdw}{\aword, p}$ holds.
Moreover, the value of $ \mysharp{<}{\shdw}{\aword, p}$ is equal to the  sum of $ \mysharp{<}{\shdw \wedge \betashdw}{\aword, p}$ and $ \mysharp{<}{\shdw \wedge \neg \betashdw}{\aword, p}$.
Hence, the above equations can be transformed~into:
\[\mysharp{<}{\wht \wedge \alpha}{\aword, p} - \mysharp{<}{\wht \wedge \beta}{\aword, p} + \mysharp{<}{\shdw}{\aword, p} = \frac{p}{2}\]
\[\mysharp{<}{\wht \wedge \alpha}{\aword, p} - \mysharp{<}{\wht \wedge \beta}{\aword, p} + \mysharp{<}{\shdw \wedge \betashdw}{\aword, p} + \mysharp{<}{\shdw \wedge \neg\betashdw}{\aword, p} = \frac{p}{2}\]
\[\mysharp{<}{\wht \wedge \alpha}{\aword, p} + \mysharp{<}{\shdw \wedge \neg\betashdw}{\aword, p} = \frac{p}{2},\]
which can be rewritten into an $\LTLFMajP{}$ formula $\phiisequal{\alpha}{\beta} \deffnoindent \Half ( [\wht \land \alpha] \vee [\shdw \land \neg \betashdw] )$, due to the disjointness of shadows and whites.
The correctness of the presented formula follows immediately from the correctness of each arithmetical transformation and the semantics of $\LTLFMajP{}$.
\end{proof}
The presented exercises show that the expressive power of $\LTLFHalf$ is so high that, under a mild assumption of \kl{truly-shadowness}, it allows 
us to perform cardinality comparison. 
From here, we are only a step away from showing undecidability of the logic, which is tackled next.

%%%%%%%%%%%%%%%%%%%%%%%%%%%%%%%%%%%%%%%%% ENCODING OF MINSKY MACHINES %%%%%%%%%%%%%%%%%%%%%%%%%%%%%%%%%%%%%%%%%%%%%%%%%%%%%%%%%%%%%%%%%%%%

\section{Undecidability of LTL extensions} \label{sec:majority}
This section is dedicated to the main technical contribution of the paper, namely that $\LTLFHalf$,  $\LTLFMajP$ and $\LTLFMostFreq$ have undecidable satisfiability and model checking problems. 
We start from $\LTLFHalf$. 
Then, the undecidability of $\LTLFMajP$ will follow immediately from the fact that $\Halfname{}$ is definable by $\MajPname{}$. 
Finally, we will show how the undecidability proof can be adjusted to $\LTLFMostFreq$.

We start by recalling the basics on Minsky Machines.

\paragraph*{Minsky machines}
A \newnotion{deterministic Minsky machine} is, roughly speaking, a finite transition system equipped with two unbounded-size natural counters, where each counter can be incremented, decremented (only in the case it is positive), and tested for being zero. 
Formally, a \kl{Minsky machine} $\minsky$ is composed of a finite set of \emph{states}~$Q$ with a distinguished \emph{initial} state $q_0$ and a~transition function $\delta: ( Q \times  \{ 0, + \}^2 ) \rightarrow ( \{ -1, 0, 1 \}^2 {\times} (Q \setminus \{q_0\})$ satisfying three additional requirements: whenever $\delta(q,f,s)=(\counterf,\counters,q')$ holds, $\counterf=-1$ implies $f={+}$, $\counters=-1$ implies~$s={+}$ (\ie it means that only the positive counters can be decremented) and $q \neq q'$ (the machine cannot enter the same state two times in a row). 
Intuitively, the first coordinate of $\delta$ describes the current state of the machine, the second and the third coordinates tell us whether the current value of the $i$-th counter is zero or positive, the next two coordinates denote the update on the counters and the last coordinate denotes the target state.

We define a \emph{run} of a \kl{Minsky machine} $\minsky$ as a sequence of consecutive transitions of~$\minsky$. 
Formally, a run of $\minsky$ is a finite word $\aword \in (Q {\times}  \{ 0, {+} \}^2 \times  \{-1,0, 1\}^2 \times Q \setminus \{q_0\})^{+}$ such that, when denoting~$\aword_i$ as $(q^i,f^i,s^i, \counterf^i, \counters^i, q^i_N)$, all the following conditions are satisfied:
\begin{enumerate}
\item\label{item:start} $q^0 = q_0$ and $f^0 = s^0 = 0$, 
\item\label{item:state-consistency} for each $i$ we have $\delta(q^i, f^i, s^i)=(\counterf^i, \counters^i, q_N^i)$,
\item\label{item:state-propagation} for each $i<|\aword|$ we have $q^i_N=q^{i+1}$,
\item\label{item:counter-consistency} \label{item:last} for each $i$, $f^i$ equals $0$ iff $\counterf^0+ \dots + \counterf^{i-1}=0$, and $+$ otherwise; similarly  $s^i$ is $0$ if $\counters^0+ \dots+ \counters^{i-1}=0$ and $+$ otherwise. 
\end{enumerate}
It is not hard to see that this definition is equivalent to the classical one~\cite{Minsky}. 
We say that a \kl{Minsky machine} \emph{reaches} a state $q \in Q$ if there is a run with a letter containing $q$ on its last coordinate. 
It is well known that the problem of checking whether a given \kl{Minsky machine} reaches a~given state is undecidable~\cite{Minsky}.
\subsection{``Half of'' meets the halting problem} \label{subsec:halting}

We start from presenting the overview of the claimed reduction.
Until the end of Section~\ref{sec:majority}, let us fix a \kl{Minsky machine} $\minsky = (Q, q_0, \delta)$ and its state $\fancyq \in Q$.
Our ultimate goal is to define an $\LTLFHalf$ formula $\psiminskyq$ such that~$\psiminskyq$ has a model iff~$\minsky$ reaches $\fancyq$. 
To do so, we define a formula $\psiminsky$ such that there is a one-to-one correspondence between the models of $\psiminsky$ and runs of $\minsky$. 
Expressing the reachability of $q$, and thus $\psiminskyq$, based on $\psiminsky$ is easy.

Intuitively, the formula $\psiminsky$ describes a \kl{shadowy} word $\aword$ encoding on its white positions the consecutive letters of a run of $\minsky$.
In order to express it, we introduce a set $\SigmaMinsky$, composed of the following distinguished atomic propositions:
\begin{itemize}
\item $\fromq{q}$ and $\toq{q}$ for all states $q \in Q$, 
\item $\cfst{c}$ and $\csnd{c}$ for counter values $c \in \{0, +\}$, and
\item $\ifst{\op}$ and $\isnd{\op}$ for all operations $\op \in \{-1, 0, 1\}$.
\end{itemize}

We formalise the one-to-one correspondence as the function $\run$, which takes an appropriately defined shadowy model and returns a corresponding run of $\minsky$. 
More precisely, the function $\run(\aword)$ returns a run whose $i$th configuration is $(q,f,s,\counterf, \counters,q_N)$ if and only if the $i$th white configuration of $\aword$ is labelled with $\fromq{q}, \cfst{f}, \csnd{s}, \ifst{\counterf}, \isnd{\counters}$ and $\toq{q_N}$.

The formula $\psiminsky$ ensures that its models are \kl{truly~$\SigmaMinsky$-shadowy} words representing a run satisfying properties P\ref{item:start}--P\ref{item:last}. To construct it, we start from $\psitrulySigmaMinsky$ and extending it with four conjuncts.
The first two of them represent properties P\ref{item:start}--P\ref{item:state-consistency} of runs. 
They can be written in $\LTLF$ in an obvious way. 

To ensure the satisfaction of the property P\ref{item:state-propagation}, we observe that in some sense the letters~$\fromq{q}$ and $\toq{q}$ are paired in a model, \ie always after reaching a state in $\minsky$ you need to get out of it (the initial state is an exception here, but we assumed that there are no transitions to the initial state). 
Thus, to identify for which $q$ we should set the $\fromq{q}$ letter on the position~$p$, it is sufficient to see for which state we do not have a corresponding pair, \ie for which state~$q$ the number of white $\fromq{q}$ to the left of $p$ is not equal to the number of white~$\toq{q}$ to the left of $p$. 
We achieve this in the spirit of Exercise~\ref{ex:3}. 

Finally, the satisfaction of the property P\ref{item:counter-consistency} can be achieved by checking for each position~$p$ 
whether the number of white $\ifst{{+}1}$ to the left of $p$ is the same as the number of white~$\ifst{{-}1}$ to the left of $p$, and similarly for the second counter.
This reduces to checking an equicardinality of certain sets, which can be done by employing shadows and Exercise~\ref{ex:3}.

\paragraph*{The reduction}
Now we are ready to present the claimed reduction. 

\noindent We first restrict the class of models under consideration to \kl{truly~$\SigmaMinsky$-shadowy} words (for the feasibility of equicardinality encoding) with a formula \newnotion{$\psitrulySigmaMinsky$}.
Then, we express that the models satisfy properties~P\ref{item:start} and~P\ref{item:state-consistency}. 
The first property can be expressed with $\psi_{P\ref{item:start}} \deff \fromq{q_0} \land \cfst{0} \land \csnd{0}$.

\noindent The property P\ref{item:state-consistency} will be a conjunction of two formulae. 
The first one, namely $\psi^1_{P\ref{item:state-consistency}}$, is an immediate implementation of P\ref{item:state-consistency}. 
The second one, \ie $\psi^2_{P\ref{item:state-consistency}}$, is not necessary, but simplifies the proof; we require that no position is labelled by more than six letters from $\SigmaMinsky$.
\[
    \psi^1_{P\ref{item:state-consistency}} \deff \G (\wht \rightarrow \hspace{-2.7em}\bigvee_{\delta(q, f, s) =  (\counterf, \counters, q_N)} \hspace{-2em}\fromq{q}\land \cfst{f} \land \csnd{s} \land \ifst{\counterf} \land \isnd{\counters} \land\toq{q_N}),
\]
\[
\psi^2_{P\ref{item:state-consistency}} \deff
\G{}\hspace{-2em}
\bigwedge_{\substack{p_1,\dots, p_7 \in \SigmaMinsky\\
p_1, \dots, p_7\text{ are pairwise different}}}
\hspace{-2em}
 \neg (p_1 \land p_2 \land \dots \land p_7).
\]
\noindent We put $\psi_{P\ref{item:state-consistency}} \deff \psi^1_{P\ref{item:state-consistency}} \land \psi^2_{P\ref{item:state-consistency}}$ and $\psi_{\textit{enc-basics}} \deff \psitrulySigmaMinsky \land \psi_{P\ref{item:start}} \land \psi_{P\ref{item:state-consistency}}$.

We now formalise the correspondence between intended models and runs. Let $\run$ be the function which takes a word $\aword$ satisfying $\psi_{\textit{enc-basics}}$ and returns the word $\mword$ such that~$|\mword| = |\aword|/2$ and for each position $i$ we have: 
\[
(\leftrightsquigarrow): 
\mword_i=(q,f,s,\counterf,\counters,q_N) \; \text{iff} \;
\aword_{2i} \supseteq \set{
    \wht,
    \fromq{q}, 
    \cfst{f}, 
    \csnd{s}, 
    \ifst{\counterf}, 
    \isnd{\counters},
    \toq{q_N}
    }.
\]

Note that the definition of $\psi_{\textit{enc-basics}}$ makes the function run correctly defined and unambiguous, and that the results of run satisfy properties P\ref{item:start} and~P\ref{item:state-consistency}.
We summarise this as the following fact.
\begin{fact} \label{fact:run}
The function $\run$ is uniquely defined and returns words satisfying~P\ref{item:start} and~P\ref{item:state-consistency}.
\end{fact}

What remains to be done is to ensure properties P\ref{item:state-propagation} and P\ref{item:counter-consistency}.
We start from the former~one. The formula $\psi_{P\ref{item:state-propagation}}$ relies on the tools established in Exercise~\ref{ex:3} and is defined as follows:
\[
\psi_{P\ref{item:state-propagation}} \deff \G{\left( \wht \rightarrow \bigwedge_{q \in Q \setminus \{ q_0 \}} (\fromq{q} \lor  \phiisequal{\fromq{q}}{\toq{q}})\right)}.
\]

\begin{lemma}\label{lemma:c3}
If $\aword$ satisfies $\psi_{\textit{enc-basics}} \wedge \psi_{P\ref{item:state-propagation}}$, then $\run(\aword)$ satisfies P\ref{item:start}--P\ref{item:state-propagation}.
\end{lemma}
\begin{proof} 
Note that the satisfaction of the properties P\ref{item:start} and P\ref{item:state-consistency} by $\run(\aword)$ follows from Fact~\ref{fact:run}.
Hence, to finish the proof it amount to show that $\run(\aword)$ satisfies the property P\ref{item:state-propagation}.

Ad absurdum, assume that $\run(\aword)$ does not satisfy P\ref{item:state-propagation}.
It implies the existence of a white position $p$ in $\aword$ such that $\aword, p \models \toq{q}$ but $\aword, p{+}2 \models \fromq{q'}$ for some~$q \neq q'$. 
By our definition of \kl{Minsky machines}, we conclude that $\aword, p \models \fromq{q''}$ for some $q'' \neq q$.
Thus, $\aword, p \not\models \fromq{q}$. 

From the satisfaction of $\psi_{P\ref{item:state-propagation}}$ by $\aword$ we know that $\aword, p \models \phiisequal{\fromq{q}}{\toq{q}}$. 
Let $k$ be the total number of positions labelled with $\fromq{q}$ before $p$.
Since $\aword, p \models \phiisequal{\fromq{q}}{\toq{q}}$ holds, by Exercise~\ref{ex:3} we infer that the number of positions satisfying $\toq{q}$ before $p$ is also equal to $k$.
Since $\aword, p{+}2 \not\models \fromq{q}$ and from the satisfaction of $\psi_{P\ref{item:state-propagation}}$ by $\aword$ we once more conclude $\aword, p{+}2 \models \phiisequal{\fromq{q}}{\toq{q}}$. 
But such a situation clearly cannot happen due to the fact that the number of $\toq{q}$ in the past is equal to $k+1$, while the number of $\fromq{q}$ in the past is $k$.
\end{proof}

\noindent Finally, to express the property P\ref{item:counter-consistency}, we once again employ the tools from Exercise~\ref{ex:3}, \ie: 
\[
\psi_{P\ref{item:counter-consistency}} \deff
\G (
  \cfst{0} \to \phiisequal{\ifst{{+}1}}{\ifst{{-}1}}  )
\; \land
 \G ( 
  \csnd{0} \to \phiisequal{\isnd{{+}1}}{\isnd{{-}1}} )
  \] 
  \[
\land
\G (\wht \to (\cfst{0} \leftrightarrow \neg\cfst{+})) \land \G(\wht \to (\csnd{0} \leftrightarrow \neg\csnd{+}))
\]
The use of $\leftrightarrow$ in $\psi_{P\ref{item:counter-consistency}}$ guarantees that $\cfst{0}$ labels exactly the white positions having the counter empty (and similarly for the second counter). 
The counters are never decreased from $0$, thus the white positions not satisfying~$\cfst{0}$ are exactly those having the first counter positive.

Finally, let us define $\psiminsky$ as $\psi_{\textit{enc-basics}} \land \psi_{P\ref{item:state-propagation}} \land \psi_{P\ref{item:counter-consistency}}$.
The proof of the forthcoming fact relies on the correctness of Exercise~\ref{ex:3} and is quite similar to the proof of Lemma~\ref{lemma:c3}.

\begin{lemma}\label{lemma:c4}
If $\aword$ satisfies $\psiminsky$, then $\run(\aword)$ is a run of $\minsky$.
\end{lemma}
\begin{proof}
Let $\aword \models \psiminsky$.
From Lemma~\ref{lemma:c3} we know that $\run(\aword)$ satisfies properties P\ref{item:start}--P\ref{item:state-propagation}.
By the definition of the run of $\minsky$, to show that $\run(\aword)$ is a run of $\minsky$, it suffices to show that $\run(\aword)$ satisfies P\ref{item:counter-consistency}.
In order to do it, we need to show that for all positions $i$ of $\run(\aword)$ we have that $f^i$ equals $0$ iff~$\counterf^0+ \dots + \counterf^{i-1}=0$, and $+$ otherwise.
Note that the ``otherwise'' part follows from the last two conjuncts of $\psi_{P\ref{item:counter-consistency}}$ and that the proof for the second counter is analogous. Hence, we omit it.

Take any white $i$. 
Our claim can be equivalently phrased as $\aword, 2i \models \cfst{0}$ holds iff $S = \sum_{j=0, (\aword, 2j) \models \ifst{\op}}^{i-1} \op$ is equal to $0$.
Note that the terms with $\op = 0$ do not contribute to the total value of $S$, so they can be omitted.
Moreover, by pushing all terms with $\op = {-}1$ to the RHS, we can represent the equation 
$S = 0$ as \[\sum_{j=0, (\aword, 2j) \models \ifst{{+}1}}^{i-1} 1 = - \sum_{j=0, (\aword, 2j) \models \ifst{{-}1}}^{i-1} -1.\]
The above equality obviously holds iff the total number of white positions before $i$ labelled with $\ifst{{+}1}$ and the total number of white positions before $i$ labelled with $ \ifst{{-}1}$ coincide.
Hence, by Exercise~\ref{ex:3}, exactly iff~$\phiisequal{\ifst{{+}1}}{\ifst{{-}1}}$ is satisfied.
But $\aword, i \models \phiisequal{\ifst{{+}1}}{\ifst{{-}1}} \leftrightarrow \cfst{0}$ holds due to the satisfaction of~$\psi_{P\ref{item:counter-consistency}}$. 
Thus, we can conclude that $\run(\aword)$ is indeed a run of $\minsky$.
\end{proof}

Lastly, to show that the encoding is correct, we need to show that each run has a corresponding model.
It is again easy: it can be shown by constructing an appropriate $\aword$; the white positions are defined according to $(\leftrightsquigarrow)$, and the shadows can be constructed~accordingly.

\begin{fact}\label{fact:c5}
If $\mword$ is a run of $\minsky$, then there is a word $\aword \models \psiminsky$ satisfying $\run(\aword)=\mword$.
\end{fact}
\begin{proof}
Take $\aword$ be a word of length $2 \cdot |\mword|$ defined as follows (for all $0 \leq i < |\mword|$):
\[
(\textit{whites}): 
\mword_i=(q,f,s,\counterf,\counters,q_N) \; \text{iff} \;
\aword_{2i} = \set{
    \wht,
    \fromq{q}, 
    \cfst{f}, 
    \csnd{s}, 
    \ifst{\counterf}, 
    \isnd{\counters},
    \toq{q_N}
    }.
\]
\[
(\textit{shadows}):
\mword_i=(q,f,s,\counterf,\counters,q_N) \; \text{iff} \;
\aword_{2i{+}1} = \set{
    \shdw,
    \widetilde{\fromq{q}}, 
    \widetilde{\cfst{f}}, 
    \widetilde{\csnd{s}}, 
    \widetilde{\ifst{\counterf}}, 
    \widetilde{\isnd{\counters}},
    \widetilde{\toq{q_N}}
    }.
\]
From the construction we see that $\aword$ satisfies $\psitrulySigmaMinsky{}$. 
Moreover, due to the $(\leftrightsquigarrow)$ correspondence and Fact~\ref{fact:run} we conclude $\aword$ satisfies $\psi_{\textit{enc-basics}}$.
Then, it is routine to check that $\aword$ satisfies $\psi_{P\ref{item:state-propagation}}$ and $\psi_{P\ref{item:counter-consistency}}$. 
Hence, $\aword \models \psiminsky$.
\end{proof}

Let $\psiminskyq \deff \psiminsky \land \F{(\toq{\fancyq})}$.
Observe that the formula~$\psiminskyq$ is satisfiable if and only if $\minsky$ reaches~$\fancyq$. 
The ``if'' part follows from Lemma~\ref{lemma:c4} and the satisfaction of the conjunct $\F{(\toq{\fancyq})}$ from~$\psiminsky$. 
The ``only if'' part follows from Fact~\ref{fact:c5}.
Hence, from undecidability of the reachability problem \kl{Minsky machines} we infer our main theorem:
\begin{thm} \label{thm:main}
The satisfiability problem for $\LTLFHalf$ is undecidable.
\end{thm}

\subsection{Undecidability of model-checking} \label{subsec:model-checking}

For a given alphabet $\Sigma$, we can define a Kripke structure $\kripkep{\Sigma}$ whose set of traces is the language $(2^\Sigma)^+$:  the set of states $S$ of $\kripkep{\Sigma}$ is composed of all subsets of $\Sigma$, all states are initial (\ie $I = S$), the transition relation is the maximal relation ($R = S {\times} S$) and $\ell(X) {=} X$ for any subset~$X \subseteq \Sigma$.
It follows that a formula $\varphi$ over an alphabet $\Sigma$ is satisfiable if and only if there is a trace of $\kripkep{\Sigma}$ satisfying $\varphi$. 
Hence, from the undecidability of the satisfiability problem for~$\LTLFHalf$ we get:
\begin{thm}
Model-checking of $\LTLFHalf$ formulae over Kripke structures is undecidable.
\end{thm}
The decidability can be regained if additional constraints on the shape of Kripke structures are imposed: model-checking of $\LTLFHalf$ formulae over \emph{flat} structures is decidable~\cite{DeckerHLST17}.

As discussed earlier, the $\Halfname{}$ operator can be expressed in terms of the $\MajPname{}$ operator. 
Hence, we conclude:
\begin{corollary}
Model-checking and satisfiability problems for $\LTLFMajP$ are undecidable.
\end{corollary}

\subsection{Most-Frequent Letter and Undecidability} \label{subsec:mfc}

We next turn our attention to the $\MostFreqname$ operator, which turns out to be a little bit problematic. 
Typically, formulae depend only on the atomic propositions that they explicitly mentioned. 
Here, it is not the case. 
Consider a formula $\varphi = \MostFreq a$ and words $\aword_1 = \setof{a}\setof{}\setof{a}$ and $\aword_2 = \setof{a,b}\setof{b}\setof{a,b}$. 
Clearly, $\aword_1, 2 \models \varphi$ whereas $\aword_2, 2 \not \models \varphi$. 
This can be fixed in many ways -- for example, by parametrising $\MostFreqname$ with a domain, so that it expresses that ``$a$ is the most frequent letter among $b_1, \ldots, b_n$''. 
We show, however, that even this very basic version of $\MostFreqname$  is undecidable.
The proof is an adaptation of our previous proofs with a little twist inside.

First, we adjust the definition of \kl{shadowy} words. 
A word $w$ is \newnotion{strongly shadowy} if $w$ is shadowy and for each even position of $\aword$ we have that $\wht$ and $\shdw$ are the most frequent letters among the other labelling $\aword$ while for odd positions $\wht$ is the most frequent.
Note that the words constructed in the previous sections were \kl{strongly shadowy} because each letter $\sigma$ appeared only at whites or at shadows.

\begin{exercise}\label{ex:psishadowy-ok-for-mfl}
There exists an $\LTLFMostFreq{}$ formula $\psishadowyMFL$ defining \kl{strongly shadowy} words.
\end{exercise}
\begin{proof}
    It suffices to revisit Exercise~\ref{ex:1} and to modify the formula $\phiodd{}$ stipulating that odd positions are exactly those labelled with $\shdw$ (since it is the only formulae employing $\Half{}$).
    We claim that $\phiodd{}$ can be expressed with~
    \[\phioddMFL \deffnoindent \G[\MostFreq(\wht) \land (\wht \leftrightarrow \MostFreq(\shdw))]\] 
    Indeed, take any word $\aword \models \phiinit \land \phioddMFL$.
    Of course we have $\aword, 0 \models \wht$ (due to $\phiinit$). 
    Moreover, $\aword, 1 \models \shdw$ holds: otherwise we would get contradiction with $\shdw$ not being the most frequent letter in the past of $1$.
    Now assume $p>1$ and assume that the word $\aword_0, \dots, \aword_{p-1}$ is \kl{strongly shadowy}. Consider two cases.
If $p$ is odd, then both $\wht$ and $\shdw$ are the most frequent letters in the past of $p{-}1$ and $p{-}1$ is labelled by $\wht$. 
        Then, $\shdw$ is not the most frequent letter in the past of $p$ and thus $p$ is labelled by $\shdw$ and $\wht$ is the most frequent letter in the past of $p$.
If $p$ is even, $p{-}2$ is labelled by $\wht$ and the most frequent letters in the past of $p{-}2$ are $\wht$ and $\shdw$, and $p{-}1$ is labelled by $\shdw$. 
        Thus both $\wht$ and $\shdw$ are the most frequent letters in the past of $p$ and therefore $\wht$ is labelled by $\wht$.        %
   Thus, $\aword_0, \dots, \aword_{p}$ is \kl{strongly shadowy}. 
   By induction, $\aword$ is \kl{strongly shadowy}.
   
   It can be readily checked that every \kl{strongly shadowy} word satisfies $\psishadowyMFL$.
\end{proof}

We argue that over the \kl{strongly shadowy} models, the formulae $\Half{\sigma}$ and $\MostFreq{\sigma}$ are equivalent.

\begin{lemma} \label{lemma:equivalence-halfish}
For all \kl{strongly shadowy} words $\aword \models \psishadowyMFL$, all even positions $2i$ and all letters $\sigma$ we have the equivalence $\aword, 2i \models \Half{\sigma}$ iff $\aword, 2i \models \MostFreq{\sigma}$.
\end{lemma}
\begin{proof}
If $\aword, 2i \models \MostFreq{\sigma}$, then 
$\aword, 2i \models \MostFreq{\wht}$ due to the \kl{strongly shadowness} of $\aword$.
Hence $\mysharp{<}{\sigma}{\aword, 2i} = \mysharp{<}{\wht}{\aword, 2i} = \frac{2i}{2}$, implying $\aword, 2i \models \Half{\sigma}$.

Now, assume that $\aword, 2i \models \Half{\sigma}$ holds, so $\sigma$ appears $i$ times in the past.
Since $\aword$ is \kl{strongly shadowy} we know that $\wht$ is the most frequent letter.
Moreover, $\wht$ appears $\frac{2i}{2} = i$ times in the past. 
Hence, $\aword, 2i \models \MostFreq{\sigma}$.
\end{proof}

We say that a letter $\sigma$ is \newnotion{importunate} in a word $\aword$ if $\sigma$ labels more than half of the positions in some even prefix of $\aword$. Notice that \kl{strongly shadowy} words cannot have \kl{importunate} letters.

With the above lemma, it is tempting to finish the proof as follows: replace each $\Half(\varphi)$ in the formulae from Section~\ref{subsec:halting} with $\MostFreq(p_\varphi)$ for some fresh atomic proposition~$p_{\varphi}$ and require that $\G(\varphi \leftrightarrow p_{\varphi})$ holds. 
A formula obtained from $\varphi$ in this way will be called a \newnotion{dehalfication} of $\varphi$ and will be denoted with $\dehalf{\varphi}$. The next lemma shows that $\dehalf{\cdot}$ preserves satisfaction of certain $\LTLFHalf{}$ formulae. 
\begin{lemma}\label{lemma:equivalence-halfish2}
Let $\varphi$ be an $\LTLFHalf$ formula without nested $\Half{}$ operators and without $\F{}$ modality, $\Lambda$ be the set of all formulae $\lambda$ such that $\Half{\lambda}$ appears in $\varphi$ and let $\aword$ be a word such that $\aword \models \psishadowyMFL \land \bigwedge_{\lambda \in \Lambda} \G(p_{\lambda} \leftrightarrow \lambda)$. 
Then for all even positions $2p$ of~$\aword$ we have that $\aword, 2p \models \dehalf{\varphi}$ implies $\aword, 2p \models \varphi$. Moreover, $\aword \models \G(\wht \rightarrow \dehalf{\varphi})$ implies $\aword \models \G(\wht \rightarrow \varphi)$.
\end{lemma}
\begin{proof}
The proof goes via structural induction over $\LTLFHalf{}$ formulae without nested $\Half{}$ operators and without $\F{}$ operators.
The only interesting case is when $\varphi = \Half{\lambda}$, which follows from Lemma~\ref{lemma:equivalence-halfish}.
\end{proof}

Note, however, that the above lemma works only one way: it fails when the formula $\varphi$ is satisfied in more than half of the positions of some prefix, as that would make $p_\varphi$ \kl{importunate} leading to unsatisfiablity of $\psishadowyMFL{}$. 

\newcommand{\psibasicencmfl}{\psi_{\textit{enc-basics}}^{\textit{MFL}}}
\newcommand{\psibasicenc}{\psi_{\textit{enc-basics}}}

\subsection{Most-Frequent Letter: the reduction} \label{subsec:mfc-reduction}

The next step is to construct a formula defining \kl{truly~$\SigmaMinsky$-shadowy} words, which are the crucial part of~$\psibasicencmfl$. 
To do it, we first need to rewrite a formula $\phitransfer{\sigma}{\sigmashdw}$, transferring the truth of a letter $\sigma$ from whites into their shadows. 
The main ingredient of $\phitransfer{\sigma}{\sigmashdw}$ is the formula $\varphi_{(\heartsuit)} \deffnoindent \G{\left( \wht \rightarrow \Half{( [\wht \land \sigma] \vee [\shdw \land \neg \sigmashdw])} \right)}$, which we replace with $\dehalf{\varphi_{(\heartsuit)}}$.
We call the obtained formula $(\phitransfer{\sigma}{\sigmashdw})^{\textit{MFL}}$ and show its correctness below.

First, by Lemma~\ref{lemma:equivalence-halfish2} we know that every model of $(\phitransfer{\sigma}{\sigmashdw})^{\textit{MFL}}$ is also a model of $\phitransfer{\sigma}{\sigmashdw}$.
Then, the models of~$\phitransfer{\sigma}{\sigmashdw}$ can be made strongly shadowy, so \kl{dehalfication} of $\phitransfer{\sigma}{\sigmashdw}$ is satisfiability-preserving.
\begin{lemma} \label{lemma:corectness-dehalf}
Let $p_{\varphi}$ be a fresh letter for $\varphi := [\wht \land \sigma] \vee [\shdw \land \neg \sigmashdw]$.
Take $\aword$, a \kl{strongly shadowy} word satisfying $\aword \models \phitransfer{\sigma}{\sigmashdw}$ without any occurrences of $p_{\varphi}$.
Then $\aword'$, the word obtained by labelling with $p_{\varphi}$ all the positions of $\aword$ satisfying $\varphi$, is \kl{strongly shadowy}.
\end{lemma}
\begin{proof}
Ad absurdum, assume that $\aword'$ is not \kl{strongly shadowy}.
Since $\aword$ is \kl{strongly shadowy}, it implies that $p_{\varphi}$ is importunate, \ie there is some even prefix $\aword'_{\leq 2i}$ of $\aword'$ in which the number of occurrences of $p_{\varphi}$ is greater than~$i$.
More precisely, we have that $\mysharp{<}{p_{\varphi}}{\aword, 2i} = \mysharp{<}{\wht \land \sigma}{\aword, 2i} + \mysharp{<}{\shdw \land \neg \sigmashdw}{\aword, 2i} > i$.
But since $\aword' \models \phitransfer{\sigma}{\sigmashdw}$, we know that $\aword' \models \varphi_{(\heartsuit)}$, which implies $\aword', 2i \models \Half{[\wht \land \sigma] \vee [\shdw \land \neg \sigmashdw]}$.
So $\mysharp{<}{\wht \land \sigma}{\aword, 2i} + \mysharp{<}{\shdw \land \neg \sigmashdw}{\aword, 2i} = i$, contradicting the previous assumption.
Hence, $\aword'$ is \kl{strongly shadowy}.
\end{proof}
Hence, we obtain the correctness of $(\phitransfer{\sigma}{\sigmashdw})^{\textit{MFL}}$.
By applying the same strategy to other conjuncts of~$\psi_{\textit{enc-basics}}$ and  Fact~\ref{fact:run}, we obtain $\psibasicencmfl$ satisfying:
\begin{corollary}
The function $\run$ (taking as arguments the words satisfying $\psibasicencmfl$) is uniquely defined and returns words satisfying~P\ref{item:start} and~P\ref{item:state-consistency}.
Moreover the formulae $\psibasicencmfl$ and $\psi_{\textit{enc-basics}}$ are equi-satisfiable.
\end{corollary}

Towards completing the undecidability proof we need to prepare the rewritings of the formulae $\psi_{P\ref{item:state-propagation}}$ and~$\psi_{P\ref{item:counter-consistency}}$.
For $\psi_{P\ref{item:state-propagation}}$ we proceed similarly to the previous case. 
We know that the models of $\psibasicencmfl \land \dehalf{\psi_{P\ref{item:state-propagation}}}$ satisfy~P\ref{item:state-propagation} (due to Lemma~\ref{lemma:equivalence-halfish2} they satisfy $\psi_{P\ref{item:state-propagation}}$ and hence, by Lemma~\ref{lemma:c3}, also P\ref{item:state-propagation}).
To observe the existence of such models, we show again that the satisfiability of $\psi_{P\ref{item:state-propagation}}$ is preserved by \kl{dehalfication}.
\begin{lemma}\label{lemma:mflc3}
Let $p_{q}$ be a fresh letter for $\varphi_q := [\wht \land \fromq{q}] \vee [\shdw \land \neg \widetilde{\toq{q}} ]$ indexed over $q \in Q \setminus \{q_0\}$.
Take $\aword$, a \kl{strongly shadowy} word satisfying $\aword \models \psibasicencmfl \land \psi_{P\ref{item:state-propagation}}$ without any occurrences of $p_{q}$.
Then $\aword'$, the word obtained by labelling with $p_{q}$ all the positions of $\aword$ satisfying $\varphi_q$, is \kl{strongly shadowy}.
\end{lemma}
\begin{proof}
Ad absurdum, assume that $\aword'$ is not \kl{strongly shadowy}.
Since $\aword$ is \kl{strongly shadowy}, it implies that some letter $p_{q}$ is \kl{importunate}, \ie there is some even prefix $\aword'_{\leq 2i}$ of $\aword'$ in which the number of occurrences of $p_{q}$ is greater than~$i$.
Hence, take $q$ and $i$ that such a prefix is the shortest one.

Note that for $p_q$ to be \kl{importunate} means that the following inequality holds:
\[ \mysharp{<}{p_q}{\aword', 2i} = \mysharp{<}{\wht \land \fromq{q}}{\aword', 2i} + \mysharp{<}{\shdw \land \neg \widetilde{\toq{q}}}{\aword', 2i} = \mysharp{<}{\wht \land \fromq{q}}{\aword', 2i} + \mysharp{<}{\shdw \land \neg \toq{q}}{\aword', 2i} = \]
\[ = \mysharp{<}{\wht \land \fromq{q}}{\aword', 2i} + (i - \mysharp{<}{\wht \land \toq{q}}{\aword', 2i})   > i,\]
which is clearly equivalent to $(\star): \mysharp{<}{\wht \land \fromq{q}}{\aword', 2i} > \mysharp{<}{\wht \land \toq{q}}{\aword', 2i}$.

We consider two cases depending on the satisfaction of $\fromq{q}$:
\begin{itemize}
    \item $\aword', 2i \models \neg \fromq{q}$.
    Then, from the satisfaction of $\psi_{P\ref{item:state-propagation}}$ we know that $\aword', 2i \models \phiisequal{\fromq{q}}{\toq{q}}$. 
    It implies the equality  $\mysharp{<}{\wht \land \fromq{q}}{\aword', 2i} = \mysharp{<}{\wht \land \toq{q}}{\aword', 2i}$ that contradicts the inequality $(\star)$.

    \item $\aword', 2i \models \fromq{q}$.
    For inequality $(\star)$ to hold it is necessary for $i$ to be positive.
    From the satisfaction $\aword \models \psibasicencmfl \land \psi_{P\ref{item:state-propagation}}$ we know that P\ref{item:state-propagation} holds.
    It implies, by the definition of a run of Minsky machine, that $\aword', 2i{-}2 \models \fromq{q'} \land \toq{q}$ for some $q' \neq q$.
    Moreover, the word $\aword_{\leq 2i{-}2}'$ does not have importunate letters and hence, we know that the inequality $\mysharp{<}{\wht \land \fromq{q}}{\aword', 2i{-}2} \leq \mysharp{<}{\wht \land \toq{q}}{\aword', 2i{-}2}$ holds.
    Note that due to the satisfaction $\aword', 2i{-}2 \models \toq{q}$ we infer the following inequality contradicting $(\star):$ \[\mysharp{<}{\wht \land \fromq{q}}{\aword', 2i{-}2} = \mysharp{<}{\wht \land \fromq{q}}{\aword', 2i} \leq \mysharp{<}{\wht \land \toq{q}}{\aword', 2i{-}2} < \mysharp{<}{\wht \land \toq{q}}{\aword', 2i{-}2} + 1 =\mysharp{<}{\wht \land \toq{q}}{\aword', 2i}.\]
\end{itemize}
Hence, $\aword'$ is \kl{strongly shadowy}.
\end{proof}

From Lemma~\ref{lemma:c3}, Lemma~\ref{lemma:mflc3} and Lemma~\ref{lemma:equivalence-halfish2} we immediately conclude:
\begin{corollary}
If $\aword$ satisfies $\psibasicencmfl \wedge \dehalf{\psi_{P\ref{item:state-propagation}}}$, then $\run(\aword)$ 
satisfies P\ref{item:start}--P\ref{item:state-propagation}. 
Moreover the formulae~$\psibasicencmfl \wedge \dehalf{\psi_{P\ref{item:state-propagation}}}$ and~$\psibasicenc \wedge \psi_{P\ref{item:state-propagation}}$ are equi-satisfiable.
\end{corollary}

The last formula to rewrite is $\psi_{P\ref{item:counter-consistency}}$.
We focus only on its first part, speaking about the first counter, \ie
\[\G ( \cfst{0} \to \Half{([\wht \land \ifst{{+}1}] \vee [\shdw \land \neg \widetilde{\ifst{{-}1}} ])} \land \G( \wht \to (\cfst{0} \leftrightarrow \neg\cfst{+}))\]
Note that this time we cannot simply dehalfise this formula: the letter responsible for the inner part of $\Halfname{}$ would necessarily be importunate -- consider an initial fragment of a run of $\minsky$ in which $\minsky$ increments its first counter without decrementing it.
Fortunately, we cannot say the same when the machine decrements the counter and hence, it suffices to express the equivalent (due to even length of shadowy models) statement $\psi'_{P\ref{item:counter-consistency}}$ as follows:\\
\[ \G ( \cfst{0} \to \Half{\neg([\wht \land \ifst{{+}1}] \vee [\shdw \land \neg \widetilde{\ifst{{-}1}} ])}  \land \G( \wht \to (\cfst{0} \leftrightarrow \neg\cfst{+})) \]
As we did before, we show that dehalfication of $\psi'_{P\ref{item:counter-consistency}}$ preserves satisfiability:
\begin{lemma} \label{lemma:corectness-mfl-counters}
Let $p_{\varphi}$ be a fresh letter for $\varphi := \neg([\wht \land \ifst{{+}1}] \vee [\shdw \land \neg \widetilde{\ifst{{-}1}} ])$.
Take $\aword$, a \kl{strongly shadowy} word satisfying $\aword \models \psibasicencmfl \wedge \dehalf{\psi_{P\ref{item:state-propagation}}} \land \psi'_{P\ref{item:counter-consistency}}$ without any occurrences of $p_{\varphi}$.
Then $\aword'$, the word obtained by labelling with $p_{\varphi}$ all the positions of $\aword$ satisfying $\varphi$, is \kl{strongly shadowy}.
\end{lemma}
\begin{proof}
Ad absurdum, assume that $\aword'$ is not \kl{strongly shadowy}.
Since $\aword$ is \kl{strongly shadowy}, it implies that $p_{\varphi}$ is importunate, \ie there is some even prefix $\aword'_{\leq 2i}$ of $\aword'$ in which the number of occurrences of $p_{\varphi}$ is greater than~$i$.
It implies that $\mysharp{<}{p_{\varphi}}{\aword', 2i} > i$. 
We can calculate that:
\[
    \mysharp{<}{p_{\varphi}}{\aword', 2i} = 2i - \mysharp{<}{\wht \land \ifst{{+}1}}{\aword', 2i} - \mysharp{<}{\shdw \land \neg \widetilde{\ifst{{-}1}}}{\aword', 2i} =
\]
\[ 
2i - \mysharp{<}{\wht \land \ifst{{+}1}}{\aword', 2i} - (i -\mysharp{<}{\shdw \land \widetilde{\ifst{{-}1}}}{\aword', 2i}) = 
\]
\[ 
    i -  \mysharp{<}{\wht \land \ifst{{+}1}}{\aword', 2i} + \mysharp{<}{\wht \land \ifst{{-}1}}{\aword', 2i}) > i,
\]
which is equivalent to the following inequality:
\[
    \mysharp{<}{\wht \land \ifst{{-}1}}{\aword', 2i} > \mysharp{<}{\wht \land \ifst{{+}1}}{\aword', 2i}
\]
We consider two cases:
\begin{itemize}
    \item When $\aword', 2i \models \cfst{0}$.
    From the satisfaction of $\aword' \models \psi'_{P\ref{item:counter-consistency}}$ we know that $\mysharp{<}{\wht \land \ifst{{-}1}}{\aword', 2i}$ is equal to~$\mysharp{<}{\wht \land \ifst{{+}1}}{\aword', 2i}$, contradicting the previously obtained inequality.

    \item When $\aword', 2i \models \cfst{+}$, then by $\aword' \models \psi'_{P\ref{item:counter-consistency}}$ we know that $\mysharp{<}{\wht \land \ifst{{-}1}}{\aword', 2i}$ is smaller than or equal to~$\mysharp{<}{\wht \land \ifst{{+}1}}{\aword', 2i}$, leading again to contradiction. 
\end{itemize}
Hence, $\aword'$ is \kl{strongly shadowy}.
\end{proof}
Finally, let $(\psiminskyq)^{\textit{MFL}} \deffnoindent \psibasicencmfl \land \dehalf{\psi_{P\ref{item:state-propagation}}} 
\land \dehalf{\psi_{P\ref{item:counter-consistency}}} \land \F{\toq{\fancyq}}$.
From Lemma~\ref{lemma:c4}, Lemma~\ref{lemma:corectness-mfl-counters} and Lemma~\ref{lemma:equivalence-halfish2} we immediately conclude:
\begin{corollary}
If $\aword$ satisfies $(\psiminskyq)^{\textit{MFL}}$ then it satisfies P\ref{item:start}--P\ref{item:counter-consistency}. 
Moreover the formulae~$(\psiminskyq)^{\textit{MFL}}$ and~$\psiminskyq$ are equi-satisfiable.
\end{corollary}

Thus, by Theorem~\ref{thm:main} and the above corollary, we obtain the undecidability of $\LTLFMostFreq{}$.
Undecidability of the model-checking problem is concluded by virtually the same argument as in Section~\ref{subsec:model-checking}.
Hence:
\begin{thm}
The model-checking and the satisfiability problems for $\LTLFMostFreq$ are undecidable.
\end{thm}

%%%%%%%%%%%%%%%%%%%%%%%%%%%%%%%%%%%%%%%%% DECIDABLE VARIANTS %%%%%%%%%%%%%%%%%%%%%%%%%%%%%%%%%%%%%%%%%%%%%%%%%%%%%%%%%%%%%%%%%%

\section{Decidable variants}

We have shown that $\LTLF$ with frequency operators lead to undecidability. 
Without the operators that can express $\Fname{}$ (\eg $\Fname{}$, $\Gname{}$ or $\U{}$), the decision problems become $\NP$-complete. 
Below we assume the standard semantics of $\LTL$ operator $\X{}$, \ie $\aword, i \models \X{\varphi}$ iff $i{+}1 < |\aword|$ and $\aword, i{+}1 \models \varphi$.

\begin{thm}
Model-checking and satisfiability problems for $\LTLp{\Xname{},\MostFreqname{},\MajPname{}}$ are $NP$-complete.
\end{thm}
\begin{proof}
Let $\varphi \in \LTLp{\Xname{},\MostFreqname{},\MajPname{}}$ be a formula of temporal depth $d$ (\ie the maximal number of nested $\X{}$ operators).
Then it is easy to see that $\aword \models \varphi$ iff $\aword_{\leq d{+}1} \models \varphi$, \ie that the only relevant part of $\aword$ required for the satisfaction of $\varphi$ are its first $d{+}1$ positions. 
Thus, to solve the satisfiability problem, it suffices to guess a word $\aword_{\leq d{+}1}$ (which is polynomial size) and to check  whether it satisfies $\varphi$ (which can be done in polynomial time by a naive evaluation algorithm).
Thus the satisfiability problem is in $\NP$.
For the model checking problem we proceed similarly. 
Note that it amounts to guessing a fragment of a trace of a Kripke structure (of length $\leq d{+}1$) and test if it satisfies $\varphi$, which again can be done in $\NP$.
The matching lower bounds are inherited from $\LTLp{\X{}}$~\cite{BaulandM0SSV11}. 
\end{proof}

The reason why the complexity of the logic $\LTLp{\Xname{},\MostFreqname{},\MajPname{}}$ is so low is that the truth of the formula depends only on some initial fragment of a trace. This is, however, a big restriction of the expressive power. 
Thus, we consider a different approach motivated by the work of~\cite{BokerCHK14}.

In the new setting, we allow to use arbitrary $\LTL$ formulae as well as percentage operators as long as the they are not mixed with $\G{}$. 
We introduce a logic $\LTLprocent$, which extends the classical $\LTL$~\cite{Pnueli77} with the percentage operators of the form~$\mathbf{P}_{\bowtie k\%}\varphi$ for any $\bowtie\; \in \set{\leq, <, =, >, \geq}$, $k \in \N$ and $\varphi \in \LTL$.
By way of example, the formula~$\mathbf{P}_{< 20\%}(a)$ is true at a position $p$ if less then $20\%$ of positions before $p$ satisfy $a$. The past majority operator is a special case of the percentage operator: $\MajP \; \equiv \; \P{}_{\geq 50\%}$.
Formally:
\begin{center}
\begin{tabular}[t]{lll}
$\aword,i \models \P{}_{\bowtie k\%}\varphi$ & \text{if} & $|\setof{j < i \colon {\aword, j} \models {\varphi}}| \bowtie \frac{k}{100} i$
\end{tabular}
\end{center}

To avoid undecidability, the percentage operators cannot appear under negation or be nested. 
Therefore, the syntax of $\LTLprocent$ is defined with the following grammar:
\[
    \varphi, \varphi' ::= 
    \psi_{\LTL{}} 
    \; \mid \;
     \varphi \lor \varphi' 
    \; \mid \;
     \varphi \land \varphi' 
      \; \mid \;
       \F (\psi_{\LTL{}} \land \mathbf{P}_{\bowtie k\%}\psi_{\LTL{}}'), 
\]
where $\psi_{\LTL{}}$, $\psi_{\LTL{}}'$ are (full) $\LTL{}$ formulae. 

\newcommand{\automaton}{\mathcal{A}}
\newcommand{\parikh}{\mathcal{P}}
\newcommand{\inequalities}{\mathcal{E}}

The main tool used in the decidability proof is the Parikh Automata~\cite{KlaedtkeR03}.
A~Parikh automaton~$\parikh = (\automaton, \inequalities)$ over the alphabet $\Sigma$ is composed of a finite-state automaton~$\automaton$ accepting words from $\Sigma^*$ and a semi-linear set $\inequalities$ given as a system of linear inequalities with integer coefficients, where the variables are $x_a$ for $a \in \Sigma$. 
We say that $\parikh$ accepts a word~$\aword$ if $\automaton$ accepts $\aword$ and the mapping assigning to each variable $x_a$ from $\inequalities$ the total number of positions of $\aword$ carrying the letter~$a$, is a solution to $\inequalities$. 
Checking non-emptiness of the language of $\parikh$ can be done in $\NP$~\cite{FigueiraL15}.

Now we proceed with our main decidability results. 
It is obtained by constructing an appropriate Parikh automaton recognising the models of an input $\LTLprocent$ formula.

\begin{thm}
The satisfiability problem for $\LTLprocent$ is decidable.
\end{thm}
\begin{proof}
Let $\varphi \in \LTLprocent$. 
By turning $\varphi$ into a DNF, we can focus on checking satisfiability of some of its conjuncts.
Hence, w.l.o.g. we assume that $\varphi = \varphi_0 \wedge \bigwedge_{i=1}^{n} \varphi_i$, where $\varphi_0$ is in $\LTL$ and all $\varphi_i$ have the form $ \F (\psi_{\LTL{}}^{i,1} \land \mathbf{P}_{\bowtie k_i\%}\psi_{\LTL{}}^{i,2})$ for some $\LTL$ formulae $\psi_{\LTL{}}^{i,1}$ and~$\psi_{\LTL{}}^{i,2}$. 
Observe that a word $\aword$ is a model of $\varphi$ iff it satisfies~$\varphi_0$ and for each conjunct $\varphi_i$ we can pick a witness position $p_i$ from $\aword$ such that $\aword, p_i \models \psi_{\LTL{}}^{i,1} \land \mathbf{P}_{\bowtie k_i\%}\psi_{\LTL{}}^{i,2}$.
Moreover, the percentage constraints inside such formulae speak only about the prefix $\aword_{< p_i}$. 
Thus, knowing the position~$p_i$ and the number of positions before $p_i$ satisfying $\psi_{\LTL{}}^{i,2}$, the percentage constraint inside $\varphi_i$ can be imposed globally rather than locally. 
It suggests the use of Parikh automata: the $\LTL$ part of $\varphi$ can be checked by the appropriate automaton $\automaton$ (due to the correspondence that for an $\LTL$ formula over finite words one can build a finite-state automaton recognising the models of such a formula~\cite{GiacomoV15}) and the global constraints, speaking about the satisfaction of percentage operators, can be ensured with a set of linear inequalities~$\inequalities$. 

Our plan is as follows: we decorate the intended models $\aword$ with additional information on witnesses, such that the witness position $p_i$ for $\varphi_i$ will be labelled by $w_i$ (and there will be a unique such position in a model), all positions before $p_i$ will be labelled by $b_i$ and, among them, we distinguish with a letter $s_i$ some special positions, \ie those satisfying $\psi_{\LTL{}}^{i,2}$. More formally, for each $\varphi_i$ we produce an $\LTL$ formula $\varphi_i'$ according to the following rules:
\begin{itemize}
    \item there is a unique position $p_i$ such that $\aword, p_i \models w_i$ (selecting a witness for $\varphi_i$),
    \item for all $j < p_i$ we have $\aword, j \models b_i$ (so the positions before $p_i$ are labelled with $b_i$),
    \item $\aword \models \G(s_i \rightarrow [b_i \wedge \psi_{\LTL{}}^{i,2}])$ (distribution of the special positions among $b_i$) and
    \item $\aword, p_i \models \psi_{\LTL{}}^{i,1}$ (a precondition for $\varphi_i$).
\end{itemize}
Let $\varphi' \deffnoindent \varphi_0 \wedge \bigwedge_{i=1}^{n} \varphi_i' \wedge \bigwedge_{i=1}^{n} \F(p_i \wedge \mathbf{P}_{\bowtie k_i\%} s_i)$.
Note that $\aword \models \varphi'$ implies $\aword \models \varphi$. Moreover, any model $\aword \models \varphi$ can be labelled with letters $b_i, s_i, w_i$ such that the decorated word satisfies~$\varphi'$.
Let $\varphi'' \deffnoindent \varphi_0 \wedge \bigwedge_{i=1}^{n} \varphi_i'$ and let $\inequalities$ be the system of $n$ inequalities with $\inequalities_i = 100 \cdot x_{b_i} \bowtie k_i \cdot x_{s_i}$. 
Now observe that any model of $\varphi'$ satisfies $\inequalities$ (i.e. the value assigned to $x_a$ is the total number of positions labelled with a), due to the satisfaction of counting operators, and vice versa: every word~$\aword \models \varphi''$ satisfying $\inequalities$ is a model of $\varphi''$. 
It gives us a sufficient characterisation of models of~$\varphi$.
Let $\automaton$ be a finite automaton recognising the models of $\varphi''$, then a Parikh automaton~$\parikh = (\automaton, \inequalities)$, as we already discussed, is non-empty if and only if $\varphi$ has a model. 
Since checking non-emptiness of $\parikh$ is decidable, we can conclude that $\LTLprocent$ is decidable.
\end{proof}
% %
A rough complexity analysis yields an $\NExpTime$ upper bound on the problem: the automaton $\parikh$ that we constructed is exponential in $\varphi$ (translating $\varphi$ to DNF does not increase the complexity since we only guess one conjunct, which is of polynomial size in $\varphi$).
Moreover, checking non-emptiness can be done non-deterministically in time polynomial in the size of the automaton. 
Thus, the problem is decidable in $\NExpTime$.
The bound is not optimal: we conjuncture that the problem is $\PSpace$-complete. 
We believe that by employing techniques similar to~\cite{BokerCHK14}, one can construct~$\parikh$ and check its non-emptiness on the fly, which should result in the $\PSpace{}$ upper bound.

\newcommand{\kripke}{\mathcal{K}}
For the model-checking problem, we observe that determining whether some trace of a Kripke structure $\kripke=(S,I,R,l)$ satisfies $\varphi$ is equivalent to checking the satisfiability of formula~$\varphi_\kripke \land \varphi$, where $\varphi_\kripke$ is a formula describing all the traces of $\kripke$.
Such a formula can be constructed in a standard manner. 
For simplicity, we treat~$S$ as a set of auxiliary letters, and consider the conjunction of (1) \( \bigvee_{s \in I}s \), (2) \( \G( \X \top \rightarrow \bigvee_{(s, s') \in R} (s \land \X s'))\) and (3) \( \bigwedge_{s \in S} \G (s \rightarrow \bigwedge_{p \in \ell(s)} p) \), expressing that the trace starts with an initial state, consecutive positions describe consecutive states and that the trace is labelled by the appropriate letters. 
Therefore, the model-checking problem can be reduced in polynomial time to the satisfiability problem.
\begin{corollary}
The model-checking problem for $\LTLprocent$ is decidable.
\end{corollary}

%%%%%%%%%%%%%%%%%%%%%%%%%%%%%%%%%%%%%%%%% FO2+MAJORITY %%%%%%%%%%%%%%%%%%%%%%%%%%%%%%%%%%%%%%%%%%%%%%%%%%%%%%%%%%%%%%%%%%%%%%%%

\section{Two-Variable First-Order Logic with Majority Quantifier} \label{subsec:fo2maj}

The \emph{Two-Variable First-Order Logic on words}, denoted here with $\FOt[<]$, is a robust fragment of First-Order Logic~$\FO$ interpreted on finite words. 
It involves quantification over variables $x$ and~$y$ (ranging over the words' positions) and it admits a linear order predicate $<$ (interpreted as a natural order on positions) and the equality predicate $=$.
Henceforth we assume the usual semantics of $\FOt[<]$ (\cf~\cite{EtessamiVW02}).

In this section, we investigate the logic $\FOtMaj[<]$, namely the extension of $\FOt[<]$ with the so-called \emph{Majority quantifier} $\Maj$. 
Such quantifier was intensively studied due to its close connection with circuit complexity and algebra, see \eg~\cite{Krebs2008, BehleK11, BehleKR09}.
Intuitively, the formula~$\Maj{x}.\varphi$ specifies that at least half of all the positions in a model, after substituting $x$ with them, satisfy~$\varphi$. 
Formally $\aword \models \Maj{x}.\varphi$ holds, if and only if $\frac{|\aword|}{2} \leq |\{ p \; \mid \; \aword, p \models \varphi[x {/} p] \}|$.
We stress that the formula $\Maj{x}.\varphi$ may contain free occurrences of the variable $y$.

Note that the Majority quantifier shares similarities to the $\MajPname{}$ operator, but in contrast to~$\MajPname{}$, the $\Maj$ quantifier counts \emph{globally}. 
We take advantage of such similarities and by reusing the technique developed in the previous sections, we show that the satisfiability problem for $\FOtMaj[<]$ is also undecidable. 
We stress that our result significantly sharpens an existing undecidability result for $\FO$ with Majority from~\cite{Lange04} (since in our case the number of variables is limited) as well as for $\FOt[<, \succ]$ with Presburger Arithmetics from~\cite{LodayaS17} (since our counting mechanism is limited and the successor relation $\succ$ is disallowed).

\subsection{Proof plan}
There are three possible approaches to proving the undecidability of $\FOtMaj[<]$. 
The first one is to reproduce all the results for $\LTLFMajP$, which is rather uninspiring. 
The second one is to define a translation from $\LTLFMajP$ to $\FOtMaj[<]$ that produces an equisatisfiable formula. 
This is possible, but because of models of odd length, it involves a lot of case study.
Here we present a third approach, which, we believe, gives the best insight: we show a translation from $\LTLFMajP$ to $\FOtMaj[<]$  that works for $\LTLFMajP$ formulae whose all models are \kl{shadowy}. 
Since we only use such models in the undecidability proof of $\LTLFMajP$, this shows the undecidability of~$\FOtMaj[<]$.

\subsection{Shadowy models}
We first focus on defining \kl{shadowy} words in $\FOtMaj[<]$.
Before we start, let us introduce a bunch of useful macros in order to simplify the forthcoming formulae. 
Their names coincide with their intuitive meaning and their~semantics.
\begin{itemize}
\item $\HalfQ{x}.\varphi \deff \Maj{x}.\varphi \wedge \Maj{x}.\neg \varphi$,
\item $\first(x) \deff \neg \exists{y} \; y<x, \; \second(x) \deff 
\exists{y} \; y<x \wedge \forall{y} \; y<x \rightarrow \first(y)$, 
\item $\last(x) \deff \neg \exists{y} \; y>x, \; \sectolast(x) \deff \exists{y} \; 
y>x \wedge \forall{y} \; y>x \rightarrow \last(y)$
\end{itemize}

The last macro ``\emph{uniquely distributes}'' letters from a finite set $\Sigma$ among the model, \ie it ensures that each position is labelled with \emph{exactly one} $\sigma$ from $\Sigma$.
\[
\distr_\Sigma \deff \forall{x} \; \bigvee_{\sigma \in \Sigma} \sigma(x) \wedge 
\bigwedge_{\sigma, \sigma' \in \Sigma, \sigma \neq \sigma'} 
\big( \neg \sigma(x) \vee \neg \sigma'(x) \big)
\]

\begin{lemma} \label{lemma:fo2-shadowy}
There is an $\FOtMaj[<]$ formula $\psishadowyfo$ defining \kl{shadowy} words.
\end{lemma} 
\begin{proof}
Let $\phibase$ be a formula defining the language of all (non-empty) words, where the letters $\wht$ and $\shdw$ label disjoint positions in the way that the first position satisfies $\wht$ and the total number of $\shdw$ and $\wht$ coincide. 
It can be written, \eg with 
$
\distr_{\{ \wht, \shdw \}} \wedge
\exists{x} (\first(x) \wedge \wht(x)) \wedge \HalfQ{x}.\wht(x) \wedge \HalfQ{x}.\shdw(x)
$.
To define \kl{shadowy} words, it would be sufficient to specify that no neighbouring
positions carry the same letter among $\set{\wht, \shdw}$. 
This can be done with, rather complicated at the first glance, formulae:
\begin{equation*}
\begin{aligned}
\phiforbidwht(x) \deff \wht(x) \rightarrow
\HalfQ{y}.\left( [y<x \wedge \wht(y)] \vee [x < y \wedge \shdw(y)] \right),
\end{aligned}
\end{equation*}
\begin{equation*}
\begin{aligned}
\phiforbidshdw(x) \deff \shdw(x) \rightarrow
\HalfQ{y}.\left([(y < x \vee x = y) \wedge \shdw(y)] \vee [x < y \wedge \wht(y)]\right).
\end{aligned}
\end{equation*}
Finally, let $\psishadowyfo \deff \phibase \wedge  \forall{x}.\left(\phiforbidwht(x) \wedge \phiforbidshdw(x)\right)$.

Showing that \kl{shadowness} implies the satisfaction of $\psishadowyfo$ can be done by routine induction. 
For the opposite direction, take $\aword \models \psishadowyfo$.
Since $\aword \models \phibase$ the only possibility for $\aword$ to not be \kl{shadowy} is to have two consecutive positions $p,p{+}1$ carrying the same letter. 
W.l.o.g assume they are both white. 
Let $w$ be the number of white positions to the left of $p$ and let $s$ be the number of shadows to the right of $p$.
By applying $\phiforbidwht$ to~$p$ we infer that $w + s = \frac{1}{2}|\aword|$.
On the other hand, by applying $\phiforbidwht$ to $p{+}1$ it follows that $(w{+}1){+}s = \frac{1}{2}|\aword|$, which contradicts the previous equation. 
Hence, $\aword$ is \kl{shadowy}. 
\end{proof}

\subsection{Translation}

It is a classical result from~\cite{EtessamiVW02} that~$\FOt[<]$ can express $\LTLF$.
We define a translation $\fottr{v}{\varphi}$ from $\LTLFMajP$ to $\FOtMaj[<]$, parametrised by a variable $v$ (where~$v$ is either $x$ or $y$ and $\bar{v}$ denotes the different variable from $v$), inductively. 
We write $v \leq \bar{v}$ rather than $v < \bar{v} \lor v = \bar{v}$ for simplicity.
For $\LTLF$ cases, we follow~\cite{EtessamiVW02}:
\begin{itemize}
\item $\fottr{v}{a} \deff  a(v)$, for a fresh unary predicate $a$ for each $a \in \APset$,
\item $\fottr{v}{\neg \varphi} \deff \neg \fottr{v}{\varphi}$, 
\item $\fottr{v}{\varphi \land \varphi'} \deff \fottr{v}{\varphi} \land \fottr{v}{\varphi'}$,
\item $\fottr{v}{\F\varphi} \deff \exists{\bar{v}} \; (v \leq \bar{v}) \wedge \fottr{\bar{v}}{\varphi}$
\item $\fottr{v}{\MajP{\varphi}} \deff \Maj{\bar{v}} ((\bar{v}<v \land \fottr{\bar{v}}{\varphi}) \lor ( \bar{v} \geq v \land \wht(\bar{v})))$.
\end{itemize}
Finally, for a given $\LTLFMajP$ formula $\varphi$, let $\fottr{}{\varphi}$ stand for $\psishadowyfo \land \exists x. (\first(x) \land \fottr{x}{\varphi})$.

The following lemma shows the correctness of the presented translation.

\begin{lemma} \label{lem:transfo2}
An $\LTLFMajP$ formula  $\varphi$ has a \kl{shadowy} model if and only if $\fottr{}{\varphi}$ has a~model.
\end{lemma}
\begin{proof}

The correctness of the translation can be shown by a induction employing the correctness of the translation from $\LTLF{}$ to $\FOt[<]$ from~\cite{EtessamiVW02}.
The only non-classical part here is the correctness of the last presented rule for the operator $\MajPname{}$. 
To do it, it suffices employ the following observation.
Consider a word~$\aword$, its position~$p$ and a formula $\varphi$.
Assume that there are $k$ positions before $p$ satisfying~$\varphi$.
Observe that~$k \geq \frac{p}{2}$ if and only if $k + \lfloor \frac{|\aword| - p}{2} \rfloor \geq \frac{|\aword|}{2}$. 
Indeed, if $p$ is even, then the above can be obtained by adding $\frac{|\aword| - p}{2}$ to both sides. 
Otherwise, $p$ is odd, and by adding $\frac{|\aword| - p}{2}$ to both sides we obtain~$k + \lfloor \frac{|\aword| - p}{2} \rfloor \geq \frac{p}{2} + \lfloor \frac{|\aword| - p}{2}\rfloor = \frac{|\aword|}{2} - \frac{1}{2}$. 
Since the left-hand side is a natural number and the right-hand side is not, we can round the latter up and obtain the required inequality.
Observe that $\lfloor \frac{|\aword| - p}{2} \rfloor$  is exactly the number of white positions that are not before $p$. 
Thus,~$k$ is at least $\frac{p}{2}$ iff $k$ plus the number of white positions that are not before $p$ is greater than or equal to $\frac{\aword}{2}$.
And that is exactly what is written as an $\FOt$ formula in the translation of $\MajPname \varphi$.
\end{proof}

Since the formulae used in our undecidability proof for $\LTLFMajP$ have only \kl{shadowy} models, by Lemma~\ref{lem:transfo2} we immediately conclude that $\FOtMaj[<]$ is also undecidable.

\begin{thm}
The satisfiability problem for $\FOtMaj[<]$ is undecidable.
\end{thm}

%%%%%%%%%%%%%%%%%%%%%%%%%%%%%%%%%%%%%%%%% CONCLUSIONS %%%%%%%%%%%%%%%%%%%%%%%%%%%%%%%%%%%%%%%%%%%%%%%%%%%%%%%%%%%%%%%%%%%%

\section{Conclusions} \label{sec:conclusions}
We have provided a simple proof showing that adding different percentage operators to $\LTLp{\F}$ yields undecidability. 
We showed that our technique can be applied to an extension of first-order logic on words, and we hope that our work will turn useful in showing undecidability for other extensions of temporal logics.
Decidability results for logics with percentage operators in restricted contexts were also provided.

%%%%%%%%%%%%%%%%%%%%%%%%%%%%%%%%%%%%%%%%% ACKNOWLEDGEMENTS %%%%%%%%%%%%%%%%%%%%%%%%%%%%%%%%%%%%%%%%%%%%%%%%%%%%%%%%%%%%%%%%%%%%

\section*{Acknowledgements}
  Bartosz Bednarczyk was supported by the Polish Ministry of Science and Higher Education program ``Diamentowy Grant'' no.~DI2017~006447. 
  Jakub Michaliszyn was supported by NCN grant no. 2017/27/B/ST6/00299.

\bibliographystyle{plain}
\bibliography{bibliography}

\end{document}